%% file: main.tex
\pdfoutput=1
\documentclass[acmsmall,10pt,screen]{acmart}

\copyrightyear{2020}

\setcopyright{rightsretained}
\acmJournal{PACMPL}
\acmYear{2020} 
\acmVolume{4} 
\acmNumber{OOPSLA} 
\acmArticle{193} 
\acmMonth{11} 
\acmPrice{}
\acmDOI{10.1145/3428261}

\input{res/packages}

\input{res/commands}
\begin{document}
\title[On the Unusual Effectiveness of Type-Aware Operator Mutations for ...]{On the Unusual Effectiveness of Type-Aware Operator Mutations for Testing SMT Solvers}

\author[Dominik Winterer]{Dominik Winterer} \affiliation{
  \institution{ETH Zurich} 
  \department{Department of Computer Science}
  \country{Switzerland}          
}
\email{dominik.winterer@inf.ethz.ch}          

\authornote{Both authors contributed equally to this work.}

\author[Chengyu Zhang]{Chengyu Zhang}
\affiliation{
  \institution{East China Normal University} 
  \department{Software Engineering Institute}
  \country{China}                               
}
\email{dale.chengyu.zhang@gmail.com}          

\authornotemark[1]

\author{Zhendong Su}
\affiliation{
  \institution{ETH Zurich}    
  \department{Department of Computer Science}
  \country{Switzerland}       
}
\email{zhendong.su@inf.ethz.ch}  

\begin{abstract}
\input{sections/abstract}
\end{abstract}

\ccsdesc[500]{Software and its engineering~Software testing and debugging}
\keywords{SMT solvers, Fuzz testing, Type-aware operator mutation}

\maketitle    

\input{res/includes}
\input{res/plots}
\input{sections/introduction}
\input{sections/illustrative-example}
\input{sections/approach.tex}
\input{sections/empirical-eval.tex}
\input{sections/bug-samples.tex}
\input{sections/related-work}
\input{sections/conclusion}
\begin{acks}
  We thank the anonymous SPLASH/OOPSLA reviewers for their valuable feedback. 
  Our special thanks go to the Z3 and CVC4 developers, especially Nikolaj Bj\o
  rner, Lev Nachmanson, Christoph M. Wintersteiger, Murphy Berzish, 
  Arie Gurfinkel, Andrew Reynolds, Andres N\"otzli, Haniel Barbosa, Clark Barrett, \emph{etc.},
  for useful information and
  addressing our bug reports. Chengyu Zhang was partially supported by the
  China Scholarship Council, NSFC Projects No.\ 61632005 and
  No.\ 61532019.
\end{acks}

\input{sections/appendix}
\bibliographystyle{ACM-Reference-Format}
\bibliography{bibliography}
\end{document}

%% file: res/packages.tex
\usepackage{graphicx}
\usepackage{wrapfig}
\usepackage{hyperref}
\usepackage{marvosym}
\usepackage{multirow}
\usepackage{scalerel}
\usepackage{wrapfig}
\usepackage[vlined]{algorithm2e}
\usepackage{caption} 
\usepackage{xcolor}
\usepackage{booktabs}
\usepackage{subcaption}
\usepackage{numprint}
\npthousandsep{,}

\usepackage{enumitem}
\usepackage{listings}
\usepackage{pgfplots}
\usepackage{microtype}
\usepgfplotslibrary{dateplot}
\pgfplotsset{compat=newest}
\usetikzlibrary{patterns}
\colorlet{shadecolor}{gray!40}

\setlist[description]{font=\normalfont\itshape}

%% file: res/commands.tex
\newcommand{\zthreestars}{5.8K\xspace}
\newcommand{\cvcfourstars}{433\xspace}
\newcommand{\reported}{\numprint{1092}\xspace}
\newcommand{\confirmed}{819\xspace}
\newcommand{\fixed}{685\xspace}
\newcommand{\withoutopts}{489\xspace} 
\newcommand{\soundnum}{184\xspace} 
\newcommand{\invmodel}{102\xspace} 
\newcommand{\zthreeconfirmed}{578\xspace} 
\newcommand{\cvcfourconfirmed}{241\xspace} 
\newcommand{\zthreeconfirmeddefault}{316\xspace} 
\newcommand{\cvcfourconfirmeddefault}{185\xspace} 
\newcommand{\zthreesoundnum}{157\xspace} 
\newcommand{\cvcfoursoundnum}{27\xspace}

\newcommand{\crashnum}{501\xspace}
\newcommand{\zthreesoundnumdefault}{114\xspace} 
\newcommand{\cvcfoursoundnumdefault}{11\xspace}

\newcommand{\toolname}{\textsf{OpFuzz}\xspace}
\newcommand{\approachname}{{type-aware operator mutation}\xspace}

\newcommand{\setofFormulas}{\ensuremath{\texttt{Seeds}}}
\newcommand{\solverOne}{\ensuremath{S_1}}
\newcommand{\solverTwo}{\ensuremath{S_2}}
\newcommand{\triggers}{triggers}

\newcommand{\numlinespythoncode}{212} 
\newcommand{\ie}{\hbox{\emph{i.e.}}\xspace}
\newcommand{\eg}{\hbox{\emph{e.g.}}\xspace}
\newcommand{\etc}{\hbox{\emph{etc.}}\xspace}

\newcommand{\etal}{\hbox{\emph{et al.}}\xspace}

\newcommand{\zcommits}{\ensuremath{377}\xspace} 
\newcommand{\cvccommits}{\ensuremath{101}\xspace}
\newcommand{\zlocchanges}{\ensuremath{34}\xspace} 
\newcommand{\cvclocchanges}{\ensuremath{63}\xspace}
\newcommand{\ztactics}{\ensuremath{126}\xspace}

\definecolor{barblack}{HTML}{cccccc}
\definecolor{bardark}{HTML}{5e5e5e}
\definecolor{linecolor}{HTML}{575757}

%% file: sections/abstract.tex
We propose \approachname, a simple, but unusually effective
approach for testing SMT solvers.  The key idea is to mutate operators
of conforming types within the seed formulas to generate well-typed
mutant formulas. These mutant formulas are then used as the test cases
for SMT solvers. We realized type-aware operator mutation within the
\toolname\ tool and used it to stress-test Z3 and CVC4, two
state-of-the-art SMT solvers. Type-aware operator mutations are
unusually effective: During one year of extensive testing
with \toolname, we reported $\reported$ bugs on Z3's and CVC4's respective GitHub issue trackers, 
out of which $\confirmed$ unique bugs were confirmed
and $\fixed$ of the confirmed bugs were fixed by the developers.
The detected bugs are highly diverse --- we found bugs of many different
types (soundness bugs, invalid model bugs, crashes, \etc), logics and
solver configurations. We have further conducted an in-depth study of 
the bugs found by \toolname. The study results show that the bugs
found by \toolname are of high quality. Many of them affect core
components of the SMT solvers' codebases, and some required major
changes for the developers to fix. Among the $\confirmed$ confirmed bugs found
by \toolname, $\soundnum$ were soundness bugs, the most critical bugs 
in SMT solvers, and $\withoutopts$ were in the default modes of the solvers. 
Notably, \toolname found $\cvcfoursoundnum$ critical soundness bugs 
in CVC4, which has proved to be a very stable SMT solver.

%% file: res/includes.tex
\newcommand{\cvcfourlogic} {
\begin{tabular}{lrrrr}

\toprule
                     \textbf{Logic} &  \textbf{S} &  \textbf{I} &  \textbf{C} &  \textbf{Total} \\
\midrule
                                            QF\_NRA &          3 &              4 &     20 &     27 \\
                                               Set &          3 &              - &     17 &     20 \\
                                             UFLIA &          - &              2 &     16 &     18 \\
                                               NRA &          1 &              1 &     15 &     17 \\
                                                UF &          1 &              - &     16 &     17 \\
                                             QF\_BV &          1 &              - &     15 &     16 \\
                                     Uncategorized &          2 &              2 &      8 &     12 \\
                                            QF\_LIA &          2 &              - &      8 &     10 \\
                                              QF\_S &          4 &              - &      6 &     10 \\
                                               LIA &          1 &              - &      8 &      9 \\
                                             QF\_UF &          1 &              - &      8 &      9 \\
                                               LRA &          - &              - &      9 &      9 \\
                                                BV &          - &              - &      8 &      8 \\
                                            QF\_LRA &          1 &              - &      5 &      6 \\
                                            QF\_NIA &          2 &              - &      3 &      5 \\
                                             QF\_AX &          - &              - &      5 &      5 \\
                                             QF\_FP &          - &              - &      4 &      4 \\
                                       QF\_AUFBVLIA &          - &              2 &      1 &      3 \\
                                         QF\_AUFLIA &          - &              3 &      - &      3 \\
                                          QF\_UFLIA &          1 &              2 &      - &      3 \\
                                          QF\_ABVFP &          - &              - &      3 &      3 \\
                                               NIA &          1 &              - &      2 &      3 \\
                                            QF\_ABV &          - &              1 &      1 &      2 \\
                                           QF\_SLIA &          1 &              - &      1 &      2 \\
                                            UFNIRA &          - &              - &      1 &      1 \\
                                              NIRA &          - &              - &      1 &      1 \\
                                           AUFNIRA &          - &              1 &      - &      1 \\
                                              UFBV &          - &              - &      1 &      1 \\
                                          QF\_UFNRA &          - &              1 &      - &      1 \\
                                          QF\_UFLRA &          - &              - &      1 &      1 \\
                                          QF\_UFIDL &          - &              - &      1 &      1 \\
                                           QF\_NIRA &          1 &              - &      - &      1 \\
                                           QF\_ALIA &          1 &              - &      - &      1 \\
                                             UFNRA &          - &              - &      1 &      1 \\
                 \hline
            \textbf{Total} &         \textbf{27} &   \textbf{19} &    \textbf{185} &    \textbf{231} \\
\bottomrule
\end{tabular}
}

\newcommand{\zthreelogic} {
\begin{tabular}{lrrrr}
\toprule
                     \textbf{Logic} &  \textbf{S} &  \textbf{I} &  \textbf{C} &  \textbf{Total} \\
\midrule
                                              QF\_S &         42 &             27 &     43 &    112 \\
                                               NRA &         20 &              - &     44 &     64 \\
                                           QF\_SLIA &         14 &              7 &     20 &     41 \\
                                            QF\_NRA &         12 &              9 &     18 &     39 \\
                                            QF\_LIA &          2 &              6 &     25 &     33 \\
                                            QF\_NIA &         13 &              3 &     16 &     32 \\
                                             UFLIA &          4 &              6 &     15 &     25 \\
                                                UF &          3 &              1 &     18 &     22 \\
                                             QF\_BV &          7 &              3 &      8 &     18 \\
                                               LIA &          4 &              - &     13 &     17 \\
                                     Uncategorized &          5 &              - &     10 &     15 \\
                                             QF\_UF &          1 &              6 &      8 &     15 \\
                                               NIA &          5 &              - &      9 &     14 \\
                                             QF\_FP &          1 &              3 &      9 &     13 \\
                                            QF\_LRA &          4 &              1 &      7 &     12 \\
                                               LRA &          4 &              1 &      7 &     12 \\
                                              Horn &          4 &              2 &      5 &     11 \\
                                             QF\_AX &          2 &              1 &      7 &     10 \\
                                              ALIA &          1 &              - &      8 &      9 \\
                                                BV &          3 &              - &      5 &      8 \\
                                          QF\_UFLIA &          - &              2 &      4 &      6 \\
                                               Set &          1 &              - &      5 &      6 \\
                                           QF\_NIRA &          1 &              4 &      1 &      6 \\
                                             UFIDL &          2 &              - &      2 &      4 \\
                                           AUFNIRA &          - &              1 &      2 &      3 \\
                                             UFLRA &          - &              - &      2 &      2 \\
                                          QF\_ABVFP &          - &              - &      2 &      2 \\
                                          QF\_UFIDL &          - &              - &      1 &      1 \\
                                             UFNIA &          1 &              - &      - &      1 \\
                                            QF\_ABV &          - &              - &      1 &      1 \\
                                                FP &          - &              - &      1 &      1 \\
                                              NIRA &          1 &              - &      - &      1 \\
                \hline
                    \textbf{Total}     & \textbf{157} &    \textbf{83} &    \textbf{316} &    \textbf{556} \\
\bottomrule
\end{tabular}
}

\newcommand{\logicdistribution}{
{\scriptsize
\begin{figure}
\zthreelogic
\hspace{0.8cm}
\cvcfourlogic
    \caption{}
\end{figure}
}
}

\newcommand{\seedtableOne}{
\begin{tabular}{lrrr}
\toprule
\textbf{Logic} &  \textbf{\# non-inc} & \textbf{\# inc} &  \textbf{\# total} \\
\midrule
QF\_SLIA    &            67,584 &           - &  67,584 \\
QF\_FP      &            40,318 &           2 &  40,320 \\
QF\_NIA     &            23,876 &          10 &  23,886 \\
AUFLIRA    &            20,011 &           - &  20,011 \\
QF\_ABVFP   &            18,093 &          69 &  18,162 \\
QF\_BVFP    &            17,231 &         182 &  17,413 \\
QF\_ABV     &            15,084 &        1,272 &  16,356 \\
UFNIA      &            13,509 &           - &  13,509 \\
QF\_NRA     &            11,489 &           - &  11,489 \\
UFLIA      &            10,137 &           - &  10,137 \\
QF\_UF      &             7,457 &         766 &   8,223 \\
QF\_DT      &             8,000 &           - &   8,000 \\
UF         &             7,668 &           - &   7,668 \\
QF\_LIA     &             6,947 &          69 &   7,016 \\
BV         &             5,846 &          18 &   5,864 \\
UFDT       &             4,527 &           - &   4,527 \\
NRA        &             3,813 &           - &   3,813 \\
QF\_UFBV    &             1,234 &        2,330 &   3,564 \\
AUFLIA     &             3,276 &           - &   3,276 \\
FP         &             2,484 &           - &   2,484 \\
LRA        &             2,419 &           5 &   2,424 \\
QF\_S       &             2,319 &           - &   2,319 \\
QF\_IDL     &             2,193 &           - &   2,193 \\
UFLRA      &               15 &        1,870 &   1,885 \\
AUFBVDTLIA &             1,708 &           - &   1,708 \\
QF\_LRA     &             1,648 &          10 &   1,658 \\
AUFNIRA    &             1,480 &         165 &   1,645 \\
QF\_AUFLIA  &             1,303 &          72 &   1,375 \\
\bottomrule
    \vspace{0.03cm}
\end{tabular}
}
\newcommand{\seedtableTwo}{
\begin{tabular}{lrrr}
\toprule
\textbf{Logic} &  \textbf{\# non-inc} & \textbf{\# inc} &  \textbf{\# total} \\
\midrule
QF\_UFLIA   &              583 &         773 &   1,356 \\
QF\_UFLRA   &             1,284 &           - &   1,284 \\
AUFDTLIA   &              728 &           - &    728 \\
LIA        &              607 &           6 &    613 \\
QF\_AX      &              551 &           - &    551 \\
QF\_UFNIA   &              478 &           1 &    479 \\
QF\_UFIDL   &              428 &           - &    428 \\
UFDTLIA    &              327 &           - &    327 \\
QF\_RDL     &              255 &           - &    255 \\
BVFP       &              224 &          10 &    234 \\
QF\_ALIA    &              126 &          44 &    170 \\
QF\_BVFPLRA &              168 &           - &    168 \\
UFBV       &              121 &           - &    121 \\
QF\_ANIA    &               95 &           5 &    100 \\
QF\_AUFBV   &               56 &          31 &     87 \\
UFIDL      &               68 &           - &     68 \\
ALIA       &               42 &          24 &     66 \\
QF\_FPLRA   &               57 &           - &     57 \\
QF\_UFNRA   &               37 &           - &     37 \\
ABVFP      &               30 &           4 &     34 \\
NIA        &               20 &           - &     20 \\
QF\_AUFNIA  &               17 &           - &     17 \\
QF\_LIRA    &                7 &           - &      7 \\
AUFNIA     &                3 &           - &      3 \\
QF\_NIRA    &                3 &           - &      3 \\
UFDTNIA    &                1 &           - &      1 \\
\hline
cvc4regr   &              176 &        1,594 &   1,770 \\
z3test     &              479 &         841 &   1,320 \\
\hline 
\textbf{Total} &   308,640   & 10,173 &  318,813 \\
\bottomrule
\end{tabular}
}

\newcommand{\comparisonConfirmed}{
{
\footnotesize
    \renewcommand{\arraystretch}{1.1}
    \begin{tabular}{lrrrr}
        \toprule
        \multirow{1}{*}[-1em]{\textbf{Approach}} & \multicolumn{2}{c}{\textbf{Bugs in Z3}} & \multicolumn{2}{c}{\textbf{Bugs in CVC4}}\\
        \cmidrule(lr){2-3} \cmidrule(lr){4-5}
        {} & soundness & all & soundness & all\\
        \midrule
        StringFuzz~\cite{blotsky-et-al-cav18} & 0 (0)  & 1 (0) & -  & - \\
        BanditFuzz~\cite{scott-et-al-cav20} & $\geq 1$ (0)  & $\geq 1$ ($0$) & - & - \\
        Bugariu \etal~\cite{bugariu2020automatically} & 3 (1) & 5 (3) & 0 (0) & 0 (0) \\
        YinYang~\cite{semantic-fusion} & 25 (24) & 39 (36) & 5 (5) & 9 (8)  \\
        STORM~\cite{mansur-etal-arxiv2020} & 21 (17) & 27 (21) & 0 (0) & 0 (0) \\
        \hline 
        \textbf{\toolname} & \textbf{\zthreesoundnum} (\zthreesoundnumdefault) & \textbf{\zthreeconfirmed} (\zthreeconfirmeddefault) & \textbf{\cvcfoursoundnum} (\cvcfoursoundnumdefault) & \textbf{\cvcfourconfirmed} (\cvcfourconfirmeddefault)    \\
        \bottomrule
    \end{tabular}
}
} 
\newcommand{\bugcounttable}{
{\footnotesize
\setlength{\tabcolsep}{0.4em}
\renewcommand{\arraystretch}{1.0}
\begin{tabular}{lrr|r}
\toprule
\textbf{Status} &  \textbf{Z3} &  \textbf{CVC4} &  \textbf{Total} \\
\midrule
Reported  &  811 &   281 &    1,092\\
Confirmed &  578 &   241 &     819 \\ 
Fixed     &  521 &   164 &     685\\
Duplicate &   88 &    18 &     106\\
Won't fix &  106 &    16 &     122\\ 
\bottomrule
\end{tabular}
}
}

\newcommand{\bugtypetable}{
{\footnotesize
\setlength{\tabcolsep}{0.4em}
\renewcommand{\arraystretch}{1.15}
\begin{tabular}{lrr|r}
\toprule
\textbf{Type} &  \textbf{Z3} &  \textbf{CVC4} &  \textbf{Total} \\
\midrule
Crash         &  316 &     185 &  501  \\ 
Soundness     &  157 &      27 &   184 \\ 
Invalid model &   83 &      19 &   102 \\ 
Others        &   22 &      10 &   32 \\
\bottomrule
\end{tabular}
}
}

\newcommand{\bugflagtable}{
{\footnotesize
\setlength{\tabcolsep}{0.4em}
\renewcommand{\arraystretch}{1.15}
\begin{tabular}{lrr|r}
\toprule
\textbf{\#Options} &  \textbf{Z3} &  \textbf{CVC4} &  \textbf{Total} \\
\midrule
default &  388 &  101 &  489 \\
1 &  109 &  67 &  176 \\
2 &  45 &  31 &   76 \\
3+ &  36 &  42 &   78 \\
\bottomrule
\end{tabular}
}
}

\newcommand{\logicplot}{

}
\newcommand{\buglogictable}{
{\small
\begin{tabular}{lrr|r}
\toprule
\textbf{Logics} &  \textbf{Z3} &  \textbf{CVC4} &  \textbf{Total} \\
\midrule
LIA     &   3 &     3 &      6 \\
LRA     &   2 &     0 &      2 \\
NRA     &  21 &     3 &     24 \\
QF\_FP   &   1 &     0 &      1 \\
QF\_NIA  &   2 &     2 &      4 \\
QF\_NRA  &   6 &     3 &      9 \\
QF\_S    &  10 &     1 &     11 \\
QF\_SLIA &   5 &     0 &      5 \\
UF      &   2 &     0 &      2 \\
UFLIA   &   6 &     1 &      7 \\
UFLRA   &   1 &     0 &      1 \\
\bottomrule
\end{tabular}
}
}

\newcommand{\buglogicplot}{
\begin{tikzpicture}[scale=1.0]
\begin{axis}[
    ybar,
    legend style={
    at={(0,0)},
    anchor=north east},
    symbolic x coords={4.5.0,4.6.0,4.7.1,4.8.1,4.8.3,4.8.4,4.8.5,trunk},
    yticklabels=\empty,
    xtick={4.5.0,4.6.0,4.7.1,4.8.1,4.8.3,4.8.4,4.8.5,trunk},
    ymin=0,
    ymax=30, 
    height=4.0cm,
    width=10cm,
    xtick pos=bottom,
    bar width=0.20cm,
    x tick label style={xshift=.2em, yshift=.15em},
    every node near coord/.append style={font=\scriptsize},
    nodes near coords,
    nodes near coords align={vertical},
    enlarge x limits=0.1,
    x tick label style={rotate=60,anchor=east},
    ytick style={draw=none}
    ]
\addplot[fill=black!80] coordinates {
    (4.5.0,8) 
    (4.6.0,5) 
    (4.7.1,5) 
    (4.8.1,5) 
    (4.8.3,5)
    (4.8.4,8)
    (4.8.5,10)  
    (trunk,24)
};
\end{axis}
\end{tikzpicture}
}

\newcommand{\mainProcess}{
\begin{algorithm}[H]
    \small
    \DontPrintSemicolon
    \SetKwFunction{randomChoice}{random.choice}
    \SetKwFunction{skeleton}{skeleton}
    \SetKwFunction{mutate}{type\_aware\_mutate}
    \SetKwFunction{val}{validate}
    \SetKwFunction{test}{\toolname}
    \SetKwProg{proc}{Procedure}{:}{}
    \proc{\test{$\setofFormulas, \solverOne, \solverTwo$, $n$}}{
        $\triggers \leftarrow \emptyset$ \;
        \While{true}{ 
            $\varphi \xleftarrow{R} \setofFormulas$ \; 
            \For{$1$ \textbf{to} $n$}{
                $\varphi' \gets$ \mutate($\varphi$) \;
                \If{$\neg$ \val{$\varphi'$, $\solverOne$, $\solverTwo$}}{ 
                    $\triggers \gets \triggers \cup \{ \varphi' \}$ \;
                }
                $\varphi \gets$ $\varphi'$ \;
            }
            \If{Interruption}{
                \textbf{break}
            }
        }
        \Return $\triggers$
    }
\end{algorithm}
}

\newcommand{\mutateFunction}{
\begin{algorithm}[H]
    \small
    \DontPrintSemicolon
    \SetKwFunction{mutate}{type\_aware\_mutate}
    \SetKwFunction{val}{validate}
    \SetKwFunction{subtypes}{subtypes}
    \SetKwProg{func}{Function}{:}{}
    \func{\mutate{$\varphi$}}{
        $f_1 \xleftarrow{R} F(\varphi)$ \;
        $f_2 \xleftarrow{R} \subtypes(f_1)$ \;
        \Return  $\varphi[f_2/f_1]$\; 
    }
\end{algorithm}
}

\newcommand{\validateFunction}{
\begin{algorithm}[H]
\small
\DontPrintSemicolon
\SetKwFunction{mutate}{mutate}
\SetKwFunction{val}{validate}
\SetKwFunction{randomChoice}{random.choice}
\SetKwProg{func}{Function}{:}{}
\func{\val{$ \varphi', \solverOne, \solverTwo$}}{
        \If{$\solverOne(\varphi') = \mathit{error} \ \vee \ $ \\
            \mbox{}\phantom{\textbf{if} \itshape}$\solverTwo(\varphi') = \mathit{error} \ \vee \ $ \\
            \mbox{}\phantom{\textbf{if} \itshape}$\solverOne(\varphi') \neq \solverTwo(\varphi')$}{
            \Return $\mathit{false}$ \;
        }
        \Return $\mathit{true}$ \; 
}
\end{algorithm}
}

\newcommand{\typetable}{
\begin{figure}[t]
    \renewcommand{\arraystretch}{1.2}
    \centering
    \footnotesize
    \begin{tabular}{lll}   
        \toprule    
        \textbf{Function types} & &\textbf{Function Symbols} \\    
        \midrule          
        $ \Gamma, A <: \top \vdash A \times \cdots \times A \to \textit{Bool}$ && \texttt{=, distinct} \\        
        $ \Gamma \vdash \textit{Quantifier} \times \textit{Bool} \to \textit{Bool}$ && \texttt{exists, forall} \\  
        $ \Gamma \vdash \textit{Bool} \times \cdots \times \textit{Bool} \to \textit{Bool}$ && \texttt{and, or, =>} \\        
        $ \Gamma, Int <: Real \vdash \textit{Real} \times \cdots \times \textit{Real} \to \textit{Bool}$  && \texttt{<=, >=, <, >} \\       
        $ \Gamma, Int <: Real \vdash \textit{Real} \times \cdots \times \textit{Real} \to \textit{Real}$ &&  \texttt{+, -, *, /} \\        
        $ \Gamma \vdash \textit{Int} \times \cdots \times \textit{Int} \to \textit{Int}$ &&  \texttt{div} \\         
        $ \Gamma \vdash \textit{Int} \times \textit{Int} \to \textit{Int}$ &&  \texttt{mod} \\               
        \bottomrule     
    \end{tabular}
\label{typetable}
\caption{SMT function symbols categorized by their type. 
    \label{fig:op-type-table}} 
\end{figure}
}

\newcommand{\filesChangesZ}{
{\scriptsize
\setlength{\tabcolsep}{0.3em}
\renewcommand{\arraystretch}{1.15}
\begin{tabular}{lrrr}
\toprule
\textbf{Filename} &  \textbf{\#LoC changes} \\
\midrule
smt/theory\_seq.cpp  &  1082 \\
ast/rewriter/seq\_rewriter.cpp &  837\\
smt/theory\_arith\_nl.h     &  637 \\
smt/theory\_lra.cpp &   434 \\
tactic/ufbv/ufbv\_rewriter\.cpp &   375 \\
math/lp/emonics.cpp &  333 \\
smt/smt\_context.cpp &   265 \\
smt/theory\_recfun.cpp &   247 \\
tactic/core/dom\_simplify\_tactic.cpp &   229 \\
tactic/arith/purify\_arith\_tactic.cpp &   224 \\
\bottomrule
\end{tabular}
}
}

\newcommand{\filesChangesCVC}{
{\scriptsize
\setlength{\tabcolsep}{0.3em}
\renewcommand{\arraystretch}{1.15}
\begin{tabular}{lrrr}
\toprule
\textbf{Filename} &  \textbf{\#LoC changes} \\
\midrule
theory/quantifiers/inst\_propagator.cpp  &  864 \\
theory/quantifiers/quantifiers\_rewriter.cpp & 611 \\
theory/arith/nonlinear\_extension.cpp   &  519\\
theory/strings/regexp\_operation.cpp  &  292\\
theory/quantifiers/local\_theory\_ext.cpp &   270 \\
theory/strings/theory\_strings.cpp &   250 \\
theory/arith/nonlinear\_extension.h &  212 \\
preprocessing/passes/int\_to\_bv.cpp &   201 \\
theory/quantifiers/inst\_propagator.h &  194 \\
smt/smt\_engine.cpp &  130 \\
\bottomrule
\end{tabular}
}
}

\newcommand{\filesCommitsZ}{
{\scriptsize
\setlength{\tabcolsep}{0.3em}
\renewcommand{\arraystretch}{1.15}
\begin{tabular}{lrrr}
\toprule
\textbf{File} &  \textbf{\#Commits} \\
\midrule
smt/theory\_seq.cpp  &  33 \\
smt/smt\_context.cpp &  30  \\
smt/theory\_lra.cpp  &  25 \\
qe/qsat.cpp & 16\\
ast/ast.cpp & 16 \\
smt/theory\_arith\_nl.h &   15 \\
ast/rewriter/seq\_rewriter.cpp & 14\\
ast/rewriter/rewriter\_def.h &  12 \\
tactic/arith/purify\_arith\_tactic.cpp & 11\\
smt/theory\_seq.h &   11\\
\bottomrule
\end{tabular}
}
}

\newcommand{\filesCommitsCVC}{
{\scriptsize
\setlength{\tabcolsep}{0.3em}
\renewcommand{\arraystretch}{1.15}
\begin{tabular}{lrrr}
\toprule
\textbf{File} &  \textbf{\#Commits} \\
\midrule
theory/arith/nonlinear\_extension.cpp  &  7 \\
theory/strings/theory\_strings.cpp & 6\\
preprocessing/passes/unconstrained\_simplifier.cpp & 5\\
theory/arith/nl\_model.cpp &  5 \\
smt/smt\_engine.cpp     &  5 \\
theory/quantifiers/extended\_rewriter.cpp &   4 \\
theory/quantifiers/quantifiers\_rewriter.cpp &   4 \\
theory/quantifiers\_engine.cpp &   3 \\
theory/quantifiers/instantiate.cpp &  3 \\
theory/arith/nonlinear\_extension.h &   3 \\
\bottomrule
\end{tabular}
}
}

%% file: res/plots.tex
\newcommand{\zthreecoverage}{
\pgfplotsset{
tick label style={font=\scriptsize},
}
\begin{tikzpicture}[scale=1.0]
\begin{axis}[
    ybar,
    height=4.0cm,
    width=5.2cm,
    enlarge y limits=upper,
    enlarge x limits={abs=0.60cm},
    tickwidth=0.1cm,
    legend columns=-1,
    legend entries={Benchmark, \toolname},
    legend style={draw=none, font=\scriptsize},
    legend to name=named,
    title={\scriptsize Z3},
    ymin=0,
    ymax=40.0,
    xtick pos=left,
    bar width=0.20cm,
    xtick distance=1,
    symbolic x coords={lines,functions,branches},
    every node near coord/.append style={font=\tiny},
    xtick=data,
    ytick={0,10,20,30,40},
    yticklabels={0\%,10\%,20\%,30\%,40\%},
    ]
\addplot[fill=white] coordinates {(lines,33.2) (functions,36.2) (branches,13.7)};
\addplot[fill=black!80] coordinates {(lines,33.5) (functions,36.4) (branches,13.8)};
\end{axis}
\end{tikzpicture}
}

\newcommand{\cvcfourcoverage}{
\pgfplotsset{
tick label style={font=\scriptsize},
}
\begin{tikzpicture}[scale=1.0]
\begin{axis}[
    ybar,
    height=4.0cm,
    width=5.2cm,
    enlarge y limits=upper,
    enlarge x limits={abs=0.60cm},
    tickwidth=0.1cm,
    legend columns=-1,
    legend entries={Benchmark, \toolname},
    legend style={draw=none, font=\scriptsize},
    legend to name=named,
    title={\scriptsize CVC4},
    ymin=0,
    ymax=50.0,
    xtick pos=left,
    bar width=0.20cm,
    xtick distance=1,
    symbolic x coords={lines,functions,branches},
    every node near coord/.append style={font=\tiny},
    xtick=data,
    ytick={0,10,20,30,40,50},
    yticklabels={0\%,10\%,20\%,30\%,40\%,50\%},
    ]
\addplot[fill=white] coordinates {(lines,28.5) (functions,47.1) (branches,14.3)};
\addplot[fill=black!80] coordinates {(lines,28.8) (functions,47.4) (branches,14.4)};
\end{axis}
\end{tikzpicture}
}

\newcommand{\zthreetrace}{
\pgfplotsset{
tick label style={font=\tiny},
}
\begin{tikzpicture}[scale=1.0]
\begin{axis}[
    title={\scriptsize Z3},
    ymin=0.2,
    height=4.5cm,
    width=13cm,
    xtick pos=left,
    ytick pos=left,
    enlarge x limits={abs=0.3cm},
    symbolic x coords={1, 2, 3, 4, 5, 6, 7, 8, 9, 10, 11, 12, 13, 14, 15, 16, 17, 18, 19, 20, 21, 22, 23, 24, 25, 26, 27, 28, 29, 30, 31, 32, 33, 34, 35, 36, 37, 38, 39, 40},
    every node near coord/.append style={font=\tiny},
    xtick={1,5,10,15,20,25,30,35,40},
    ytick={0.20,0.21,0.22,0.23,0.24,0.25,0.26,0.27,0.28,0.29,0.30,0.31,0.32,0.33},
    yticklabels={0.20,0.21,0.22,0.23,0.24,0.25,0.26,0.27,0.28,0.29,0.30,0.31,0.32,0.33},
    ]
\addplot[black!80, mark=x] coordinates {
    (1,0.325272004)
    (2,0.291573586)
    (3,0.282566156)
    (4,0.271245927)
    (5,0.262258725)
    (6,0.261269051)
    (7,0.256455079)
    (8,0.257741102)
    (9,0.255891848)
    (10,0.253493734)
    (11,0.250277328)
    (12,0.250194992)
    (13,0.248247478)
    (14,0.245925348)
    (15,0.243820131)
    (16,0.243722177)
    (17,0.244364596)
    (18,0.241686784)
    (19,0.242976892)
    (20,0.239335752)
    (21,0.240918673)
    (22,0.240913426)
    (23,0.24042387)
    (24,0.239802827)
    (25,0.240337392)
    (26,0.238957599)
    (27,0.238165576)
    (28,0.237972416)
    (29,0.238588981)
    (30,0.238776983)
    (31,0.239385663)
    (32,0.237639806)
    (33,0.237451728)
    (34,0.236910402)
    (35,0.238396427)
    (36,0.238074437)
    (37,0.237271563)
    (38,0.236727566)
    (39,0.23754841)
    (40,0.23705653)
};
\end{axis}
\end{tikzpicture}
}

\newcommand{\cvcfourtrace}{
\pgfplotsset{
tick label style={font=\tiny},
}
\begin{tikzpicture}[scale=1.0]
\begin{axis}[
    title={\scriptsize CVC4},
    ymin=0.43,
    height=4.5cm,
    width=13cm,
    xtick pos=left,
    ytick pos=left,
    enlarge x limits={abs=0.3cm},
    symbolic x coords={1, 2, 3, 4, 5, 6, 7, 8, 9, 10, 11, 12, 13, 14, 15, 16, 17, 18, 19, 20, 21, 22, 23, 24, 25, 26, 27, 28, 29, 30, 31, 32, 33, 34, 35, 36, 37, 38, 39, 40},
    every node near coord/.append style={font=\tiny},
    xtick={1,5,10,15,20,25,30,35,40},
    ytick={0.40,0.42,0.44,0.46,0.48,0.50,0.52,0.54,0.56,0.58,0.60,0.62,0.64},
    yticklabels={0.40,0.42,0.44,0.46,0.48,0.50,0.52,0.54,0.56,0.58,0.60,0.62,0.64},
    ]
\addplot[black!80, mark=x] coordinates {
(1,0.628238745)
(2,0.568436029)
(3,0.545787185)
(4,0.533950687)
(5,0.518138115)
(6,0.506677601)
(7,0.49850781)
(8,0.494197684)
(9,0.482598155)
(10,0.480448713)
(11,0.474269849)
(12,0.47042468)
(13,0.466873053)
(14,0.47053023)
(15,0.465533585)
(16,0.463928929)
(17,0.462599027)
(18,0.462313188)
(19,0.461171922)
(20,0.462343466)
(21,0.459809286)
(22,0.457312859)
(23,0.456080735)
(24,0.455002176)
(25,0.452045366)
(26,0.451411947)
(27,0.450194363)
(28,0.450372331)
(29,0.448040012)
(30,0.447894633)
(31,0.445399303)
(32,0.443597097)
(33,0.444684662)
(34,0.446493317)
(35,0.445382969)
(36,0.444070156)
(37,0.444578588)
(38,0.444378905)
(39,0.444893659)
(40,0.443383219)
};
\end{axis}
\end{tikzpicture}
}

\newcommand{\commitsFilesZ}{
    \pgfplotstableread[row sep=\\,col sep=&]{
        interval & count \\
        1&185\\
        2&81\\
        3&46\\
        4&19\\
        5&22\\
        6&7\\
        7&4\\
        8&5\\
        9&1\\
        11&3\\
        13&1\\
        17&1\\
        55&1\\
        65&1\\
        }\mydata    
\begin{tikzpicture}
    \tikzset{every node}=[font=\tiny\sffamily]
    \begin{axis}[
        ybar,
        bar width=.20cm,
        width=.82\textwidth,
        height=4.0cm,
        legend style={at={(0.5,1)},
        anchor=north,legend columns=-1},
        symbolic x coords={1,2,3,4,5,6,7,8,9,11,13,17,55,65},
        xtick pos=left,
        ytick pos=left,
        xtick=data,
        enlarge x limits={abs=0.4cm},
        nodes near coords,
        nodes near coords align={vertical},
        ymin=0,ymax=250,
        ylabel={\#Commits},
        xlabel={\#File changes},
        label style={font=\footnotesize},
        ]
        \addplot[fill=black!80] table[x=interval,y=count]{\mydata};
    \end{axis}
\end{tikzpicture}
}

\newcommand{\commitsFilesCVC}{
    \pgfplotstableread[row sep=\\,col sep=&]{
        interval & count \\
        1&59\\
        2&24\\
        3&6\\
        4&6\\
        5&3\\
        6&1\\
        13&1\\
        18&1\\
        }\mydata
    \begin{tikzpicture}
        \tikzset{every node}=[font=\tiny\sffamily]
        \begin{axis}[
            ybar,
            bar width=.20cm,
            width=0.83\textwidth,
            height=4.0cm,
            legend style={at={(0.5,1)},
            anchor=north,legend columns=-1},
            symbolic x coords={1,2,3,4,5,6,13,18},
            xtick pos=left,
            ytick pos=left,
            xtick=data,
            enlarge x limits={abs=0.4cm},
            nodes near coords,
            nodes near coords align={vertical},
            ymin=0,ymax=80,
            ylabel={\#Commits},
            xlabel={\#File changes},
            label style={font=\footnotesize},
            ]
            \addplot[fill=white] table[x=interval,y=count]{\mydata};
        \end{axis}
    \end{tikzpicture}
}

\newcommand{\commitsLoCZ}{
    \pgfplotstableread[row sep=\\,col sep=&]{
        interval & count \\
        0-10&163\\
        10-20&63\\
        100-200&18\\
        20-30&44\\
        200-300&7\\
        30-40&21\\
        300-400&2\\
        40-50&16\\
        400-500&1\\
        50-60&9\\
        500+&1\\
        60-70&13\\
        70-80&10\\
        80-90&7\\
        90-100&2\\
        }\mydata
    \begin{tikzpicture}
        \tikzset{every node}=[font=\tiny\sffamily]
        \begin{axis}[
            ybar,
            bar width=.20cm,
            width=1\textwidth,
            height=4.2cm,
            xticklabel style={rotate=60},
            legend style={at={(0.5,1)},
                anchor=north,legend columns=-1},
            symbolic x coords={0-10,10-20,20-30,30-40,40-50,50-60,60-70,70-80,80-90,90-100,100-200,200-300,300-400,400-500,500+},
            xtick pos=left,
            ytick pos=left,
            xtick=data,
            enlarge x limits={abs=0.4cm},
            nodes near coords,
            nodes near coords align={vertical},
            ymin=0,ymax=200,
            ylabel={\#Commits},
            xlabel={\#LOC changes},
            label style={font=\footnotesize},
            ]
            \addplot[fill=black!80] table[x=interval,y=count]{\mydata};
        \end{axis}
    \end{tikzpicture}
}

\newcommand{\commitsLoCCVC}{
    \pgfplotstableread[row sep=\\,col sep=&]{
    interval & count \\
    0-10&34\\
    10-20&25\\
    100-200&2\\
    20-30&9\\
    200-300&4\\
    30-40&5\\
    300-400&3\\
    40-50&4\\
    50-60&3\\
    500+&2\\
    60-70&2\\
    70-80&3\\
    80-90&3\\
    90-100&2\\
    }\mydata
    \begin{tikzpicture}
        \tikzset{every node}=[font=\tiny\sffamily]
        \begin{axis}[
            ybar,
            bar width=.20cm,
            width=1\textwidth,
            height=4.2cm,
            legend style={at={(0.5,1)},
                anchor=north,legend columns=-1},
            symbolic x coords={0-10,10-20,20-30,30-40,40-50,50-60,60-70,70-80,80-90,90-100,100-200,200-300,300-400,500+},
            xtick pos=left,
            ytick pos=left,
            xtick=data,
            xticklabel style={rotate=60},
            enlarge x limits={abs=0.4cm},
            nodes near coords,
            nodes near coords align={vertical},
            ymin=0,ymax=42,
            ylabel={\#Commits},
            xlabel={\#LOC changes},
            label style={font=\footnotesize},
            ]
            \addplot[fill=white] table[x=interval,y=count]{\mydata};
        \end{axis}
    \end{tikzpicture}
}

\newcommand{\formulasize}{
    \pgfplotstableread[row sep=\\,col sep=&]{
    interval & count \\
    0-100&154\\
    100-200&270\\
    1K-2K&24\\
    200-300&122\\
    2K-3K&11\\
    300-400&59\\
    3K-4K&8\\
    400-500&35\\
    4K-5K&1\\
    500-600&20\\
    5K+&5\\
    600-700&11\\
    700-800&14\\
    800-900&8\\
    900-1K&5\\
    }\mydata

    \begin{tikzpicture}
        \tikzset{every node}=[font=\tiny\sffamily]
        \begin{axis}[
            ybar,
            bar width=.20cm,
            width=.75\textwidth,
            height=4.0cm,
            legend style={at={(0.5,1)},
                anchor=north,legend columns=-1},
            symbolic x coords={0-100,100-200,200-300,300-400,400-500,500-600,600-700,700-800,800-900,900-1K,1K-2K,2K-3K,3K-4K,4K-5K,5K+},
            xtick=data,
            xtick pos=left,
            ytick pos=left,
            xticklabel style={rotate=60},
            enlarge x limits={abs=0.6cm},
            nodes near coords,
            nodes near coords align={vertical},
            ymin=0,ymax=320,
            ylabel={\#Formulas},
            xlabel={Formula sizes (bytes)},
            label style={font=\footnotesize},
            ]
            \addplot[fill=black!80] table[x=interval,y=count]{\mydata};
        \end{axis}
    \end{tikzpicture}
}

\newcommand{\logicsdistribution}{
\pgfplotsset{
tick label style={font=\scriptsize},
}

\begin{tikzpicture}[scale=1.0]
\begin{axis}[
    ybar,
    legend style={
    at={(0,0)},
    anchor=north east},
    symbolic x coords={NRA, UFLIA, UF, LIA, Horn, LRA, Array, NIA, BV, Set, UFLRA, NIRA, Set logic, UFNIA, Other},
    xtick={NRA, UFLIA, UF, LIA, Horn, LRA, Array, NIA, BV, Set, UFLRA, NIRA, Set logic, UFNIA,Other},
    yticklabels=\empty,
    ymin=0,
    ymax=65,
    height=4.0cm,
    width=15.2cm,
    xtick pos=left,
    bar width=0.24cm,
    x tick label style={xshift=.2em, yshift=.15em},
    every node near coord/.append style={font=\tiny},
    nodes near coords,
    nodes near coords align={vertical},
    enlarge x limits=0.1,
    x tick label style={rotate=45,anchor=east},
    ytick style={draw=none}]
\addplot[fill=bardark] coordinates {
        (NRA,52)
        (UFLIA,37)
        (UF,15)
        (LIA,14)
        (Horn,8)
        (LRA,16)
        (Array,11)
        (NIA,7)
        (BV,12)
        (Set,7)
        (UFLRA,2)
        (NIRA,1)
        (Set logic,1)
        (UFNIA,1)
};
\addplot[fill=barblack] coordinates {
        (NRA,41)
        (UFLIA,27)
        (UF,12)
        (LIA,9)
        (Horn,8)
        (LRA,8)
        (Array,7)
        (NIA,7)
        (BV,6)
        (Set,3)
        (UFLRA,2)
        (NIRA,1)
        (Set logic,1)
        (UFNIA,1)
};

\end{axis}
\end{tikzpicture}
}

%% file: sections/introduction.tex
\section{Introduction}
Satisfiability Modulo Theory (SMT) solvers are important tools for many programming 
languages advances and applications, \eg, symbolic execution~\cite{godefroid-et-al-pldi-2005, cadar-et-al2008}, 
program synthesis~\cite{solar2008program}, solver-aided programming~\cite{torlak2014lightweight}, 
and program verification~\cite{detlefs-et-al2005, deline-et-al2005}. Incorrect results 
from SMT solvers can invalidate the results of these tools, which can be disastrous 
in safety-critical domains.  Hence, the SMT community has undertaken great efforts to make SMT solvers 
reliable. Examples include the standardized input/output file formats for SMT solvers, 
semi-formal logic/theory specifications, extensive benchmark repositories, and 
yearly-held SMT solver competitions. To date, there are several mature SMT solvers, among which 
Z3~\citep{moura-bjoerner-tacas08} (\zthreestars stars on GitHub) and CVC4~\citep{barrett-cav11} 
(\cvcfourstars stars on GitHub) are the most prominent ones. Both Z3 and CVC4 are very stable and reliable. 
In Z3, there have been fewer than $150$ reported soundness bugs in more than three years, while fewer than $50$ 
in CVC4 in more than $8$ years.\footnote{Data recorded prior to any SMT fuzzing campaigns:
July 2010 to October 2019 for CVC4; April 2015 to October 2019 for Z3.} 
Despite this, SMT solvers are complex pieces of software and inevitably still have latent bugs.
Various automated testing approaches~\citep{brummayer2009fuzzing,blotsky-et-al-cav18,bugariu2018automatically} have 
been devised for finding bugs in SMT solvers.   
However, nearly all SMT solver soundness bugs have still been exposed directly by their 
client applications, not by these techniques. This has only begun to change with the recently proposed
Semantic Fusion~\citep{semantic-fusion} and STORM~\cite{mansur-etal-arxiv2020}. 
Both exposed a number of soundness bugs in Z3, while Semantic Fusion additionally exposed some soundness bugs in CVC4.
Yet, it is unclear whether SMT solvers have reached a strong level of maturity 
and how many latent bugs remain in them.
{\parindent0pt \paragraph{\textbf{Type-aware operator mutation}} To answer this question, we introduce \approachname, a simple, yet
unusually effective approach for stress-testing SMT solvers. Its key idea  
is to mutate functions within SMT formulas with functions of conforming types. Fig.~\ref{fig:illustration}
illustrates type-aware operator mutation on an example formula.    
We replace the "\texttt{\small distinct}" in $\varphi$ by an operator of conforming type, 
\eg, the equals operator "\texttt{=}" to obtain formula $\varphi_{\text{test}}$.  
We then differentially test SMT solvers with $\varphi_{\text{test}}$ as input and 
observe their results. If the results differ, \eg, one SMT solver returns \texttt{\small sat}   
while the other returns \texttt{\small unsat}, we have found a soundness bug in either   
of the tested solvers. 
Formula $\varphi_{\text{test}}$ is clearly unsatisfiable, as $b$ cannot exist whenever $a$ is odd.
In fact, while CVC4 correctly returns \texttt{\small unsat} on $\varphi_{\text{test}}$, Z3 incorrectly 
reports \texttt{\small sat} on $\varphi_{\text{test}}$. Thus, $\varphi_{\text{test}}$ 
has triggered a soundness bug in Z3 which was promptly fixed by Z3's main developer.
}
%
\begin{figure}[t]
\begin{subfigure}{.40\textwidth}
\centering
\begin{lstlisting}[basicstyle=\scriptsize\ttfamily,numbers=none]
; \phi
(assert (forall ((a Int)) 
        (exists ((b Int)) 
        (<@\colorbox{gray!40}{distinct}@> (* 2 b) a))))
(check-sat)
\end{lstlisting}
\end{subfigure}
\hspace{0.1cm}
\begin{subfigure}{.40\textwidth}
\centering
\begin{lstlisting}[basicstyle=\scriptsize\ttfamily,numbers=none]
; \phi_{test}
(assert (forall ((a Int)) 
        (exists ((b Int)) 
        (<@\colorbox{gray!40}{=}@> (* 2 b) a))))
(check-sat)
\end{lstlisting}
\end{subfigure}

\caption{Type-aware operator mutation illustrated. We mutate the distinct operator in $\varphi$ 
        to the equals operator (see $\varphi_{\text{test}}$). Formula $\varphi_{\text{test}}$ triggers 
        a soundness bug in Z3 which reports \texttt{\small sat} on this unsatisfiable formula.
        \\ \footnotesize \url{https://github.com/Z3Prover/z3/issues/3973}
    \label{fig:illustration}
        }
\end{figure}
{\parindent0pt \paragraph{\textbf{Bug hunting with \toolname}}
We have engineered \toolname, a practical realization of \approachname. 
\toolname is unusually effective. During our bug hunting campaign from September 2019  
to September 2020, we found and reported \reported bugs in Z3 and CVC4 issue trackers, among which \confirmed   
were confirmed by the developers and \fixed were already fixed. We have found bugs across    
various logics such as (non-)linear integer and real arithmetic, uninterpreted functions, 
bit-vectors, strings, sets, sequences, array, floating-point, and combinations of these       
logics. Among these, most of the bugs (\withoutopts) were found in the default modes of the solvers,      
\ie, without additionally supplied options. This underpins the importance of our findings.
We have found many high-quality soundness bugs in Z3 and notably also in CVC4,     
which has been proven to be a very robust SMT solver by previous work. 
The root causes of the bugs that we found are often complex         
and, sometimes require the developers to perform major code changes to fix the underlying issues.  
The developers of Z3 and CVC4 greatly appreciated our bug finding effort with comments like 
"Great find!", "Thanks a lot for the bug report!" or labelling our bug reports 
as "major".     
}

{\parindent0pt \paragraph{\textbf{Comparison of \toolname with recent SMT fuzzers}}
In this paper, we use the term bug trigger to refer to a formula that triggers a bug 
in an SMT solver and the term bug to refer to a single unique bug in an SMT solver.  
Note, multiple bug triggers can be caused by the same underlying bug. All bug counts 
mentioned in this paper refer to unique bugs.  
Fig.~\ref{fig:comparison} shows a comparison of \toolname with recent tools on 
SMT solver testing from the last two years in terms of bug counts. 
We have not considered older approaches       
and defer to the related work section (Section~\ref{sec:related-work}) for a detailed discussion and literature review. 
Recent approaches can be roughly separated into two categories: generators (StringFuzz~\cite{blotsky-et-al-cav18}, Bugariu and M\"uller's approach~\cite{bugariu2020automatically},
BanditFuzz~\cite{scott-et-al-cav20}) and mutators (StringFuzz, YinYang~\cite{semantic-fusion}, STORM~\cite{mansur-etal-arxiv2020}).  StringFuzz 
is a string formula generator that also comes with a mutator. It mainly targets performance 
issues in z3str3~\cite{berzish-et-al-2017}, an alternative string solver in Z3. 
StringFuzz can find correctness bugs as a by-product; the paper mentioned one. Bugariu and M\"uller's 
approach is a formula synthesizer for string logics generating formulas that are by 
construction (un)satisfiable. They found $5$ bugs in Z3 in total with $3$ soundness  
bugs, but none in CVC4. Recently, BanditFuzz, a reinforcement learning-based fuzzer  
has been proposed. Similar to StringFuzz, BanditFuzz's main focus is on performance     
issues in SMT solvers and less on correctness bugs. The authors have 
identified inconsistent results for 1600 syntactically different bug triggers on 
the four SMT solvers Z3, CVC4, MathSAT~\cite{cimatti-et-al2013}, Colibri, and 
100 bug triggers in z3str3. However, the number of unique bugs in 
Z3 remains unclear as the authors did not reduce and report the bug triggers to filter out the duplicates.
Among the mutation-based fuzzers, YinYang is an approach to stress-test SMT solvers 
by fabricating fused formula pairs that are by construction either (un)satisfiable. 
YinYang found $39$ bugs in Z3 and $9$ in CVC4. Another 
recent approach is STORM which is based on a three-phase process of seed 
fragmentation, formula generation and instance generation. STORM has found $27$ 
bugs in Z3 with $21$ being soundness bugs, but none in CVC4.  
}
\begin{figure}[t]
    \comparisonConfirmed
    \caption{Comparison of confirmed bugs found by \toolname against the bug findings of recent 
             SMT solver testing approaches on the trunks of Z3 and CVC4. 
             In parentheses: confirmed bugs in the default modes of the solvers.  
             Performance issues are excluded from this comparison.
             \label{fig:comparison}
    }
    \vspace{-0.5cm}
\end{figure}

As Fig.~\ref{fig:comparison} illustrates, our realization \toolname of type-aware operator mutation 
compares favorably against all existing approaches by a significant margin --- \toolname found 
substantially more bugs in both Z3 and CVC4 in terms of all bugs, the soundness bugs in Z3 and CVC4, and bugs 
for the default modes of the solvers. Existing approaches also extensively tested Z3 and CVC4, and have
missed the bugs found by \toolname. In summary, we make the following contributions in this paper: 
%
\begin{itemize}
    \item We introduce \approachname, a simple, but unusually effective approach for stress-testing SMT solvers;
    \item We have realized \approachname \ within our tool \toolname \ in no more than $\numlinespythoncode$ 
    lines of code. \toolname \ helps SMT solver developers and practitioners 
    to stress-test SMT solver decision procedures regardless of the used logic 
    and solver;
    \item Between September 2019 and September 2020, we stress-tested 
    Z3 and CVC4 using \toolname, and reported \reported unique bugs   
    on the respective GitHub issue trackers of Z3 and CVC4. Out of these,   
    \confirmed bugs were confirmed, and \fixed bugs were fixed. 
    Most confirmed bugs were triggered in the default modes (\withoutopts) of the solvers, 
    and many were soundness bugs (\soundnum)\footnote{All links to the bug reports 
    are provided under the URL \url{testsmt.github.io/opfuzz_bugs.html}}; 
    \item We have conducted an in-depth analysis of the bugs to understand in which parts of 
    the SMT solvers these bugs occur. Furthermore, we examined the effort  
    necessary for the developers to fix \toolname's bugs. Our results show      
    that many bugs occur in the core parts of the SMT solvers, and some      
    require the developers to perform major code changes. 
\end{itemize}
{\parindent0pt \paragraph{\textbf{Organization of the paper}}
Section~\ref{sec:illustrative-examples} illustrates the idea behind \approachname. 
Section~\ref{sec:approach} presents \approachname \ formally,
and shows how we apply it to SMT solver testing through 
our realization \toolname. We then present our empirical evaluation (Section~\ref{sec:empirical-evaluation})
which includes detailed statistics about our bug findings. Section~\ref{sec:bug-study}       
introduces our quantitative analysis of the detected bugs and in-depth investigation 
on a set of sampled bugs to provide further insight. Finally, we survey related work (Section~\ref{sec:related-work}) 
and conclude (Section~\ref{sec:conclusion}). 
}

%% file: sections/illustrative-example.tex
\begin{figure}[t]
\begin{subfigure}{.45\linewidth}
\centering
\begin{lstlisting}[basicstyle=\scriptsize\ttfamily]
(set-logic NRA)
(declare-fun a () Real)
(declare-fun b () Real)
(assert (<@\colorbox{gray!40}{=}@> (* a b) 1))
(check-sat) (get-model)   
\end{lstlisting}
\caption{Invalid model bug in CVC4.
\\ \scriptsize \url{https://github.com/CVC4/CVC4/issues/3407}\ \ 
\label{ex:c}}
\end{subfigure}
\hspace{0.5cm}
\begin{subfigure}{.45\linewidth}
\centering
\begin{lstlisting}[basicstyle=\scriptsize\ttfamily]
(set-logic NRA)
(declare-fun a () Real)
(declare-fun b () Real)
(assert (<@\colorbox{gray!40}{>}@> (* a b) 1))
(check-sat) (get-model)   
\end{lstlisting}
\caption{Mutating the equals operator in Fig.~\ref{ex:c} to a greater operator makes the bug disappear.
 \\ \scriptsize \phantom{\url{https://github.com/Z3Prover/z3/issues/xxxx}}
\label{ex:d}
}
\end{subfigure}

\vspace{0.2cm}

\begin{subfigure}{.45\textwidth}
\centering
\begin{lstlisting}[basicstyle=\scriptsize\ttfamily]
(declare-fun f (Int) Bool)
(declare-fun g (Int) Bool)
(assert (<@\colorbox{gray!40}{distinct}@> f g))
(check-sat)
(get-model)
\end{lstlisting}
\caption{ 
Invalid model bug in CVC4.\\
\scriptsize \url{https://github.com/CVC4/CVC4/issues/3527} \ \ \label{ex:g}}
\end{subfigure}
\hspace{0.5cm}
\begin{subfigure}{.45\textwidth}
\centering
\begin{lstlisting}[basicstyle=\scriptsize\ttfamily]
(declare-fun f (Int) Bool)
(declare-fun g (Int) Bool)
(assert (<@\colorbox{gray!40}{=}@> f g))
(check-sat)
(get-model)
\end{lstlisting}
\caption{Mutating the equals operator in Fig.~\ref{ex:g} to a distinct operator 
        makes the bug disappear. \\
        }
\label{ex:h}
\end{subfigure}

\vspace{0.2cm}

\begin{subfigure}{.45\textwidth}
\centering
\begin{lstlisting}[basicstyle=\scriptsize\ttfamily]
(declare-fun x () Real)                                                         
(assert (<@\colorbox{gray!40}{distinct}@> x (sin 4.0)))                                                 
(check-sat)                                                                     
\end{lstlisting}
\caption{CVC4 crashes on this formula. 
\\ \scriptsize \url{https://github.com/CVC4/CVC4/issues/3614}\ \ \label{ex:e}}
\end{subfigure}
\hspace{0.5cm}
\begin{subfigure}{.45\textwidth}
\centering
\begin{lstlisting}[basicstyle=\scriptsize\ttfamily]
(declare-fun x () Real)                                                         
(assert (<@\colorbox{gray!40}{>=}@> x (sin 4.0)))                                                 
(check-sat)              
\end{lstlisting}
\caption{Mutating the distinct operator in Fig.~\ref{ex:e} to a greater than 
operator makes the bug disappear.}
\end{subfigure}
%

\caption{Left column: bug-triggering formulas in SMT-LIB format. 
    Right column: formulas that were transformed from the corresponding
    bug-triggering formulas by a single operator change. 
    \label{fig:bug-array}}
\end{figure}

\section{Motivating Examples}
\label{sec:illustrative-examples}
This section gives a short introduction on the SMT-LIB~\cite{BarST-SMT-10} language,
the standard for describing SMT formulas, and then motivates our technique 
\emph{type-aware operator mutation}. 

{\parindent0pt\paragraph{\textbf{SMT-LIB language}}
We consider the following subset of statements: 
\texttt{\small declare-fun}, \texttt{\small assert}, \texttt{\small check-sat} 
and \texttt{\small get-model}. Variables are declared as zero-valued functions. 
For example, the declaration \texttt{\small "(declare-fun a () Real)"}
declares a variable of type real. An \texttt{\small assert} statement specifies 
constraints. The predicates within the constraints have different types, \eg,    
the constraint \texttt{\small "(assert~(<= (/ x 4)~ (*~5~x)))"} includes predicates of   
real and boolean types. Multiple constraints can be viewed as the conjunction of the 
constraints in each individual constraint statement. The \texttt{\small (check-sat)} 
statement queries the solver to decide on the satisfiability of a formula.  
If all constraints are satisfied, the formula is satisfiable; otherwise, the formula 
is unsatisfiable. We can obtain a model, \ie, a satisfiable assignment, 
for a satisfiable formula by the \texttt{\small (get-model)} statement.
}
{\parindent0pt \paragraph{\textbf{Type-aware operator mutation}}
We first examine three exemplary bugs that were found by our technique. 
Consider the formula in Fig.~\ref{ex:c} on which CVC4 returns the following model: 
$a = -\frac{3}{2}$ and $b = -\frac{1}{2}$. This model is invalid as $a \cdot b \neq 1$. 
Mutating the equals operator \texttt{\small =} to the greater operator \texttt{\small >} 
hides this bug (see Fig.~\ref{ex:d}). As another example, consider the formula 
in Fig.~\ref{ex:g}. CVC4 gives an invalid model on this formula by setting 
$f = g = \mathit{false}$. Furthermore CVC4 crashes on the formula in Fig.~\ref{ex:e}. 
Again in both cases, the bug disappears with a single operator change 
(see Fig.~\ref{ex:d} and Fig.~\ref{ex:h}).  
All illustrated cases show that operators play an important role in triggering 
SMT solver bugs. This inspired our technique, \approachname, that is to stress-test 
SMT solvers via mutating operators and use the so mutated formulas for stress-testing 
SMT solvers.
}

However, substituting an operator with another 
arbitrary operator may not always yield a syntactically correct formula.
As an example, consider Fig.~\ref{formula:original} that presents a syntactically correct seed formula. 
By substituting the first operator greater than operator \texttt{\small>} to \texttt{\small *}, 
the formula becomes syntactically incorrect (see Fig.~\ref{formula:wrong}). 
This is because the \texttt{\small assert} 
statement expects a boolean expression, while \texttt{\small *} returns a real. 
\begin{figure}[t]
\begin{subfigure}{0.45\textwidth}
\begin{lstlisting}[basicstyle=\scriptsize\ttfamily]
(declare-fun a () Real)
(assert (<@\colorbox{gray!40}{>}@> (/ (* 2 a) a) (* a a) 1))
(check-sat)
\end{lstlisting}
\vspace{-5pt}
\caption{Original formula}
\vspace{5pt}
\label{formula:original}
\end{subfigure}
\begin{subfigure}{0.45\textwidth}
\begin{lstlisting}[basicstyle=\scriptsize\ttfamily]
(declare-fun a () Real)
(assert (<@\colorbox{gray!40}{*}@> (/ (* 2 a) a) (* a a) 1))
(check-sat)
\end{lstlisting}
\vspace{-5pt}
\caption{Syntactically incorrect mutant.}
\vspace{5pt}
\label{formula:wrong}
\end{subfigure}
\begin{subfigure}{0.45\textwidth}   
\begin{lstlisting}[basicstyle=\scriptsize\ttfamily]
(declare-fun a () Real)
(assert (<@\colorbox{gray!40}{=}@> (/ (* 2 a) a) (* a a) 1))
(check-sat)
\end{lstlisting}
\vspace{-5pt}
\caption{Syntactically correct mutant.}
\label{formula:one}
\end{subfigure}
\begin{subfigure}{0.45\textwidth}   
\begin{lstlisting}[basicstyle=\scriptsize\ttfamily]
(declare-fun a () Real)
(assert (= (/ (* 2 a) a) (<@\colorbox{gray!40}{/}@> a a) 1)) 
(check-sat)
\end{lstlisting}
\vspace{-5pt}
\caption{Bug triggering mutant.}
\label{formula:two}
\end{subfigure}
\caption{Motivating examples for \approachname.}
\label{fig:motivating}
\vspace{-0.3cm}
\end{figure}
The formula of Fig.~\ref{formula:wrong} is of little value to testing an SMT solver's decision procedures  
since the solvers would reject such formulas already at a preprocessing stage.      
Hence, we have to consider the operator types for the substitutions, \ie,
avoid substituting an operator returning a boolean value, such as \texttt{\small =}, by  
an operator returning a real, such as \texttt{\small *}; neither should we substitute
an operator with a single argument, like \texttt{\small not}, by an operator 
of two or more arguments, such as \texttt{\small =}. Instead we mutate the operators 
in a \emph{type-aware} fashion. Consider the first 
\texttt{\small >} of the formula in Fig.~\ref{formula:original}. It takes an arbitrary   
number of numeral arguments and returns a boolean. Candidates for its 
substitution are \texttt{\small {<=, >=, <, =,}} and \texttt{\small{distinct}}, all of which have a conforming type, \ie, read more than 
one numerals and return a boolean. Therefore, we can safely substitute \texttt{\small >} 
of the formula in Fig.~\ref{formula:original} with a random candidate, \eg, \texttt{\small =}. 
As a result, we obtain the mutant formula in Fig.~\ref{formula:one}. This formula  
is syntactically correct and can successfully pass the preprocessing phase of the SMT solvers.
We call such mutations \emph{type-aware operator mutations}. 
As we have the guarantee that the mutant is a type-correct formula, 
we can do iterative type-aware operator mutations. Given the mutant formula in 
Fig.~\ref{formula:one}, we further substitute the second occurrence of \texttt{\small *} 
with \texttt{\small /} safely.
This yields the formula in Fig.~\ref{formula:two} which triggered a soundness bug in Z3. 
Division by zero terms are specified in the Real and Int theories of the SMT-LIB as the uninterpreted terms,
meaning that for a term \texttt{\small (/ t 0)} and arbitrary value \texttt{v}, the equation 
\texttt{\small (= v (/ t 0))} is satisfiable. In fact, we can set $a=0$ to realize a model 
for the formula in Fig.~\ref{formula:two}, \ie, let the division by zero terms be $1$ to satisfy the assert. 
Hence, the formula in Fig.~\ref{formula:two} is satisfiable. However, 
Z3 incorrectly reports \texttt{\small unsat} on it.

%% file: sections/approach.tex
\section{Approach}
\label{sec:approach}
In this section, we formally introduce \approachname and propose \toolname, 
a fuzzer for stress-testing SMT solvers. %
{\parindent0pt \paragraph{\textbf{Background}}
We consider first-order logic formulas of the satisfiability modulo 
theories (SMT). 
Such a formula $\varphi$ is satisfiable if there is at least one assignment  
on its variables under which $\varphi$ evaluates to true. Otherwise, $\varphi$ is unsatisfiable.  
We consider formulas to be realized by SMT-LIB programs\footnote{We consider the SMT-LIB language in version 2.6.} 
in which operators correspond to functions of the SMT theories. 
We use the terms functions and operators interchangeably unless stated otherwise.
}

For a formula $\varphi$, we define $F(\varphi)$ to be $\varphi$'s set of (enumerated) function occurrences. 
For example, for $\varphi= \texttt{\small (+ (* 1 1) (- 2 (* 5 2)))}$, we have 
$F(\varphi) = \{\texttt{\small +}_1, \texttt{\small*}_1, \texttt{\small-}_1, \texttt{\small*}_2\}$. 
$\varphi[f_1/f_2]$ describes the substitution of function $f_2$ by $f_1$ in $\varphi$.   
Expressions and functions are typed. 
For example, \texttt{1} is of type \textit{Int}, 
\texttt{1.0} is of type \textit{Real}, \texttt{"foo"} is of type \textit{String}.
Similarly functions also have types.  We denote the type of a function $f$ by  $f: A \to B$ 
where $A$ is the type of its arguments and $B$ its return type. We use $\Gamma$  
to denote the static typing environment of the SMT-LIB language. For example, we 
write $\Gamma \vdash \textit{Int} \times \textit{Int} \to \textit{Int}$ for the  
type of function \texttt{mod} and $\Gamma \vdash \textit{Int} \times \cdots 
\times \textit{Int} \to \textit{Int}$ for the function \texttt{div}.
$\textit{Int} \times \cdots \times \textit{Int}$ means function \texttt{div} accepts more 
than one argument with type \textit{Int}.
Fig.~\ref{fig:op-type-table} shows selected functions and their types considered in this paper. 
We emphasize that our theory is not restricted to the functions used in Fig.~\ref{fig:op-type-table}. 
It can be extended to a richer set of functions and types according to the SMT-LIB standard.  
\typetable

Similar to other programming languages with types, we can define a subtyping relation for the SMT-LIB language. 
We now formalize a fragment of the SMT-LIB's type system.  
We define type \textit{Int} to be a subtype of \textit{Real}, \ie, $\Gamma \vdash Int <: Real$.
Let \textit{A} be an arbitrary type, then we define type $\textit{A} \times \textit{A}$ 
to be a subtype of $\textit{A} \times \cdots \times \textit{A}$, \ie, 
$\Gamma, A <: \top \vdash \textit{A} \times \textit{A} <: \textit{A} \times \cdots \times \textit{A}$.
For two functions $f_1:A_1 \to B_1$ and $f_2:A_2 \to B_2$ with $A_1 <: A_2$ and 
$B_2 <: B_1$ we have 
    $$ \small \dfrac{A_1 <: A_2 \quad B_2 <: B_1}{f_2: A_2 \to B_2 <: f_1: A_1 \to B_1}$$
For example, consider function $\texttt{div}$ of type $\textit{Int} \times \cdots \times \textit{Int} \to \textit{Int}$ 
and function $\texttt{mod}$ of type $ \textit{Int} \times \textit{Int} \to \textit{Int}$ 
from Fig.~\ref{fig:op-type-table}. 
We can hence conclude that $\texttt{\small div}$'s type is a subtype of $\texttt{\small mod}$'s type:
$$ \small \dfrac{\Gamma \vdash \textit{Int} \times \textit{Int} <: \textit{Int} \times \cdots \times \textit{Int}
\quad \Gamma \vdash \textit{Int} <: \textit{Int}}
{\Gamma \vdash \textit{Int} \times \cdots \times \textit{Int} \to \textit{Int} <: \textit{Int} \times \textit{Int} \to \textit{Int}}$$
We call a formula $\varphi$ \emph{well-typed} if it complies with the rules of SMT-LIB's typing system.

\subsection{Type-Aware Operator Mutation}
Having provided basic background, we present \approachname, the key concept of this paper.  
We first introduce type-aware operator mutations and then show that type-aware operator mutants  
realize well-typed SMT-LIB programs.
\begin{definition}[Type-aware operator mutation] 
    \label{def:type-aware-mutation}
    Let $\varphi$ be an SMT formula and let $f_1:t_1$ and $f_2:t_2$ be two of its functions.  
    We say formula $\varphi' =  \varphi[f_2/f_1]$ is a \textbf{\emph{type-aware operator mutant}} 
    of $\varphi$ if $t2 <: t1$.
    Transforming $\varphi$ to $\varphi[f_2/f_1]$ is called \textbf{\emph{type-aware operator mutation}}.
\end{definition}
\begin{proposition}
    \label{well-typed-prop}
    Type-aware operator mutants are well-typed. 
\end{proposition}
\begin{proof}
    Let $\varphi$ be a well-typed SMT formula and let $\varphi'$ be a type-aware operator   
    mutant of $\varphi$. According to Definition~\ref{def:type-aware-mutation} 
    we know that $\varphi' =  \varphi[f_2/f_1]$ where $f_1:t_1$ and $f_2:t_2$ are two 
    of $\varphi$'s functions. By Definition~\ref{def:type-aware-mutation}, we also know $t2 <: t1$. 
    This implies that all arguments of $f_1$ are also accepted by $f_2$  
    and all values returned by $f_2$ could be produced by $f_1$. 
    Thus, $f_2$ accepts all the inputs provided by $\varphi'$, and formula $\varphi'$ 
    accepts all the outputs of $f_2$. Therefore $\varphi'$ is well-typed.
\end{proof}

\begin{example}
Consider $\varphi=\texttt{\small (assert (= (mod 1 1) 1)}$ with 
$F(\varphi)= \{\texttt{\small mod},\texttt{\small =}\}$. We randomly 
pick function \texttt{\small mod} from $F$, substitute it with a function that 
has its subtype, \eg, function \texttt{\small div}. We get the following type-aware operator mutant 
$\varphi'= \texttt{\small (assert (= (div 1 1) 1)}$. As Proposition~\ref{well-typed-prop} shows, 
$\varphi'$ is guaranteed to be well-typed. Thus, we can use $\varphi'$ to test SMT solvers. 
\end{example}

\begin{figure*}[t]
    \begin{subfigure}{0.49\textwidth}
    \mainProcess
    \end{subfigure}
    \begin{subfigure}{0.49\textwidth}
    \mutateFunction
    \validateFunction
    \end{subfigure}
    \caption{Left: \toolname's main process.
    Right: Function $\texttt{type\_aware\_mutate}$ realizes type-aware operator mutations
    and function $\texttt{validate}$ differentially tests the SMT solvers $S_1$ and $S_2$.}
    \label{fig:op-fuzz-realization}
\end{figure*}

\subsection{\toolname}
We implemented \toolname, a type-aware operator mutation-based fuzzer, for stress-testing
SMT solvers. \toolname leverages type-aware operator mutation to generate test inputs and validates 
the results of the SMT solvers via differential testing, 
\ie, by comparing the results of two or more SMT solvers and reporting their inconsistencies.
Fig.~\ref{fig:op-fuzz-realization} presents the algorithm of \toolname.
\toolname takes a set of seed formulas $\setofFormulas$, 
two SMT solvers $\solverOne$, $\solverTwo$ and parameter $n$ as its input. 
\toolname \ collects bug triggers in the set $\triggers$ which is initialized to the empty set. 
The main process runs inside a while loop until an interrupt is detected, \eg, by the user 
or by a time or memory limit that is reached. We first choose a random formula $\varphi$ 
from the set of formulas $\setofFormulas$ for initialization. 
In the while loop, we then perform a type-aware mutation on $\varphi$ 
realized by the \texttt{type\_aware\_mutate} function. 
In the \texttt{type\_aware\_mutate} function, we first 
randomly pick a function $f_1$ from the set of functions in $\varphi$. 
Then, we randomly choose a function $f_2$ from the set of $f_1$'s subtypes. 
The subtype function is realized based on Fig.~\ref{fig:op-type-table}.
After we obtained $\varphi' = \varphi[f_2/f_1]$ by type-aware mutation on $\varphi$, we call 
the function \texttt{validate}. It tests two SMT solvers $S_1$ and $S_2$ via differential testing 
on the input formula $\varphi'$. First, it checks whether either of the solvers has 
produced an error on processing $\varphi'$, e.g., the SMT solver did not terminate successfully, 
threw out an error message. We distinguish two cases: either $\varphi'$ triggered 
an assertion violation or segmentation fault (crash), or a model validation error that occurs  
for solvers with model validation enabled (invalid model). In both cases, the function 
returns $\mathit{false}$.   
Otherwise, it checks whether the results of the solvers are different, and returns 
$\mathit{false}$ if so, else \texttt{validate} returns $\mathit{true}$ indicating 
that $\varphi'$ has not exposed a bug trigger in neither of the solvers $S_1$ and $S_2$. 
\toolname realizes an $n$-times repeated type-aware operator mutation on every seed formula.
The choice for parameter $n$ is arbitrary but an $n$ within $200$ and $400$ has worked  
well in practice. 

\toolname is very light-weight. 
We realized \toolname in a total of only $\numlinespythoncode$ lines of Python 3.7 code. 
\toolname \ can be run in parallel mode, which can significantly increase its throughput. 
Users can customize \toolname's command-line interface to test specific solvers and/or 
configurations. \toolname can be used with any SMT solver that takes SMT-LIB v2.6 
files as its input. We have implemented the type-aware operator mutations \emph{w.r.t.} 
the function symbols in Fig.~\ref{fig:op-type-table}.  

%% file: sections/empirical-eval.tex
\section{Empirical Evaluation}
\label{sec:empirical-evaluation}
This section details our extensive evaluation with \toolname demonstrating the     
practical effectiveness of \approachname for testing SMT solvers. Between 
September 2019 and September 2020, we were running \toolname to stress-test the SMT solvers 
Z3~\cite{moura-bjoerner-tacas08} and CVC4~\cite{barrett-cav11}. We have chosen 
Z3 and CVC4, since they (1) both are popular and widely used in academia and industry,  
(2) support a rich set of logics, and (3) adopt an open-source development model. 
During our testing period, we have filed numerous bugs on the issue trackers of Z3 
and CVC4. This section describes the outcome of our efforts. 
{\parindent0pt \paragraph{\textbf{Result summary and highlights}} In summary, \toolname is unusually effective.
\begin{itemize}
    \setlength\itemsep{0.2em}
    \item \emph{Many confirmed bugs:} In one year, we have reported \reported bugs, 
    and \confirmed unique bugs in Z3 and CVC4 have been confirmed by the developers.        
    \item \emph{Many soundness bugs:} Among these, there were \soundnum soundness 
        bugs in Z3 and CVC4. Most notably, we have found \cvcfoursoundnum
        in CVC4.
    \item \emph{Most logics affected:} Our bug findings affect most SMT-LIB logics    
    including strings, (non-)linear integer and real arithmetic, bit-vectors, uninterpreted functions,   
    floating points, arrays, sets, sequences, horn, and combinations thereof. 
    \item \emph{Most bugs in default modes:} \withoutopts out of our confirmed \confirmed   
    bugs are in the default modes of the solvers, \ie, without additionally supplied    
    options. 
    %
\end{itemize}
}
\subsection{Evaluation Setup}

{\parindent0pt \paragraph{\textbf{Hardware setup and test seeds}}
We have run \toolname \ on an AMD Ryzen Threadripper 2990WX processor with 32 
cores and 32GB RAM on an Ubuntu 18.04 64-bit. 
As test seeds, we have mainly used the SMT-LIB benchmarks\cite{SMT-LIB}. 
We chose the SMT-LIB benchmarks as our test seeds since they make the largest collection of SMT 
formulas in the SMT-LIB 2.6 language.  These SMT-LIB benchmarks are also used in 
the SMTComp, the annual SMT solver competition.
Therefore, they are unlikely to trigger bugs in Z3 and CVC4 since they have already been run on them.
In addition to the SMT-LIB benchmarks, we used the regression test suites of Z3\cite{Z3test}
and CVC4\cite{CVC4test}. 
We show the seed formula counts categorized by logic and solving mode    
in the Appendix~\ref{appendix}.
We treated all seed files equally during fuzzing.  The effort spent on testing for a specific 
logic is therefore proportional to the number of its seed files within the overall seed set. 
Consequently, logics with a high seed count get tested more frequently as compared 
to others with a lower seed count.
We regularly ran Z3 and CVC4 on all seed files and excluded bug triggering seeds 
but have very rarely encountered any bug triggering seed formulas.
}%

{\parindent0pt \paragraph{\textbf{Tested options and features}}
We mainly focused our testing efforts on the default modes of the solvers. 
For CVC4, this includes enabling the options \texttt{\small --produce-models}, 
\texttt{\small --incremental} and \texttt{\small --strings-exp} as needed to 
support all test seed formulas. To detect invalid model bugs, we have supplied     
\texttt{\small --check-models} to CVC4 and \texttt{\small model.validate=true} to 
Z3. We consider these to be part of the default mode for the two solvers Z3 and CVC4  
if apart from these necessary options, no other options or tactics were used. 
Besides the default modes of Z3 and CVC4, we have considered many frequently used options      
and solver modes for Z3 and CVC4 of which we only detail a subset here. For Z3, 
we have stress-tested several tactics and several arithmetic solvers including
\texttt{\small smt.arith\_solver=x} with $\texttt{x} \in \{1,\cdots,6\}$.
We have also tested, among others, the string solver z3str3 by supplying 
\texttt{\small smt.string\_solver=z3str3}. In CVC4 we have tested, among many other options, 
syntax-guided synthesis procedure \cite{reynolds-et-al-cav2015} by specifying 
\texttt{\small --sygus-inference} and higher-order reasoning for uninterpreted 
functions by specifying \texttt{\small --uf-ho}.
}

{\parindent0pt \paragraph{\textbf{Bug types}} We have encountered many different kinds of bugs and issues     
while testing SMT solvers. We distinguish them by the following categories with  
two SMT solvers $S_1$ and $S_2$. 
\begin{itemize}
\item \emph{Soundness bug}: Formula $\varphi$ triggers a soundness bug   
    if solvers $S_1$ and $S_2$ both do not crash and give different satisfiabilities for $\varphi$.     
\item \emph{Invalid model bug}: Formula $\varphi$ triggers an invalid model bug if the model returned by the solver does not satisfy $\varphi$.
\item \emph{Crash bug}: Formula $\varphi$ triggers a crash bug if the solver throws out an assertion violation or a segmentation fault while solving $\varphi$.
\end{itemize}
\toolname \ detects soundness bug triggers by comparing the standard outputs of the solvers  
$\solverOne$ and $\solverTwo$. \toolname \ detects invalid model bug triggers by internal 
errors when using the SMT solver's model validation configuration. A crash bug trigger
is detected whenever a solver returns a non-zero exit and no timeout occurred.
}

{\parindent0pt\paragraph{\textbf{Bug trigger de-duplication}} \toolname collects bug triggers that may stem from the same underlying bug.  
Hence, we de-duplicated the bug triggers after each fuzzing run with OpFuzz to avoid 
duplicate bug reports on the GitHub issue trackers. 
Crash bugs are either assertion violations or segmentation faults. We de-duplicate assertion 
violations via the location information (file name and line number) printed on standard 
output/error. We de-duplicate segmentation faults by comparing their ASAN traces.
For soundness and invalid model bugs we used the following procedure. 
We first categorize the bug triggers by theory. We do this because bug triggers 
in different theories are likely to be unique bugs. 
Then, we select one bug trigger per theory at a time for reporting. If the bug was fixed, we checked the remaining bug triggering formulas of the same theory. 
If either one of them still triggered a bug in the solver, we repeat this process until none of them triggers a bug anymore. 
}

{\parindent0pt\paragraph{\textbf{Bug reduction}}
When a bug trigger is selected in the trigger de-duplication, we reduce the bug-triggering formula to a small enough size 
for reporting. We use C-Reduce~\citep{Regehr:2012:TRC:2254064.2254104}, a C code 
reduction tool which also works for the SMT-LIB language. 
%
}

\subsection{Evaluation Results}
Having defined the setup and bug types, we continue with the presentation of the 
evaluation results. The section is divided into three parts: (1) statistics on the bug findings by  
\toolname to assess its effectiveness, (2) coverage measurements of \toolname 
relative to the seed formulas (3) solver trace comparisons to gain further insights  
into the technique.
\begin{figure}[t]
    \begin{subfigure}[b]{0.30\textwidth}
        \centering
        \bugcounttable
        \caption{}
        \label{fig:bugcount-table}
    \end{subfigure}
    \hspace{0.5cm}
    \begin{subfigure}[b]{0.30\textwidth}
        \centering
        \bugtypetable
        \caption{\label{fig:bugtype-table}}
      \end{subfigure}
    \hspace{0.5cm}
    \begin{subfigure}[b]{0.3\textwidth}
        \centering
        \bugflagtable
        \caption{\label{fig:flag-table}}
    \end{subfigure} 
    \caption{a) Status of bugs found in Z3 and CVC4. b) Bug types among the confirmed 
    bugs. c) Number of options supplied to the solvers among the confirmed bugs.
    }
\end{figure}
{\parindent0pt\paragraph{\textbf{Bug findings}}
Fig.~\ref{fig:bugcount-table} shows the bug status counts. 
By "Reported", we refer to the unique bugs after bug trigger de-duplication that we posted on the 
GitHub issue trackers of the solvers; by "Confirmed", we refer to those posted bugs that were       
confirmed by the developers as unique bugs; by "Fixed", we refer to those posted bugs that were  
confirmed by the developers as unique bugs and addressed through at least one bug fixing commit; 
by "Duplicate", we refer to those bugs posted on GitHub that have been identified by the  
developers as duplicate to another bug report of ours or to a previously existing bug report; 
by "Won't fix", we refer to those posted bugs that were rejected by the developers, due to 
misconfigurations. 
}

We have reported a total of 
$\reported$ bugs on Z3's and CVC4's respective issue trackers. Among these, \confirmed 
unique bugs were confirmed and \fixed were fixed. Although we devoted equal testing effort to 
both solvers, we found more than twice as many bugs in Z3 as in CVC4. Previous 
approaches made similar observations~\cite{semantic-fusion}.
%
%
%
Fig.~\ref{fig:bugtype-table} shows the bug types. 
Among the bug types of the confirmed bugs, crash bugs were most frequent (\crashnum), 
followed by soundness bugs (\soundnum) and invalid model bugs (\invmodel). 
The type "Others" refers to all other unexpected behaviors in SMT solvers  
such as rejecting syntax-correct formulas, alarming invalid models when generating a valid model.
The large majority (\withoutopts out of \confirmed) of bugs found by \toolname were found 
in the default modes of the solvers, \ie, no additional options were supplied, some were  
found with one or two additional options enabled, and clearly less bugs with more than 
three options enabled (see Fig.~\ref{fig:flag-table}). 

We have also examined the distribution 
of logics among the confirmed bugs of Z3 and CVC4 (see Fig.~\ref{fig:zthreedistribution} 
and ~\ref{fig:cvcfourdistribution}). We observe that most soundness bugs in  
Z3 are in the string logics QF\_S (42), QF\_SLIA (14) and nonlinear logics    
NRA (20), QF\_NRA (12). Notably, there are also a number of soundness bugs in 
bitvectors QF\_BV (7) and linear real and integer arithmetic QF\_LRA (4), LRA (4),   
LIA (4). Similar to Z3, most soundness bugs in CVC4 are also in the string logic 
QF\_S (4) and nonlinear arithmetic QF\_NRA (3). Moreover, there are three soundness bugs   
in set logics. However, in difference to Z3 there are almost no soundness bugs in 
bitvectors and only a single soundness bug in the linear arithmetic QF\_UFLIA.     

%
{\scriptsize
\begin{figure}
    \begin{subfigure}[b]{0.40\textwidth}
        \centering
        \zthreelogic
        \caption{Z3
        \label{fig:zthreedistribution}}
    \end{subfigure}
    \hspace{0.6cm}
    \begin{subfigure}[b]{0.40\textwidth}
        \centering
        \cvcfourlogic
        \caption{CVC4
        \label{fig:cvcfourdistribution}}
    \end{subfigure}
    \caption{Logic distribution of the confirmed bugs: (S) soundness bugs, (I) 
     invalid model bugs, (C) crash bugs. "Uncategorized" refers to bugs that could 
     not be associated with either of SMT-LIB's logics.}
\end{figure}
}

\renewcommand\arraystretch{1.4}
{
\footnotesize 
\begin{figure*}[t]
\centering
\begin{tabular}{ccccccccccccc}
\toprule
&\multicolumn{5}{c}{\textbf{Z3}} & & \multicolumn{5}{c}{\textbf{CVC4}} \\ 
\cmidrule{2-6} \cmidrule{8-12}
&\emph{lines} && \emph{functions} && \emph{branches} && \emph{lines} && \emph{functions} && \emph{branches}\\
$\setofFormulas_{1000}$ &33.2\% && 36.2\% && 13.7\% &&  28.5\% && 47.1\% && 14.3\%\\
\texttt{Type-aware Mutation} &\textbf{33.5\%} && \textbf{36.4\%} &&\textbf{13.8\%} &&\textbf{28.8\%} &&\textbf{47.4\%} &&\textbf{14.4\%}\\
\bottomrule
\end{tabular}
\caption{Line, function and branch coverage achieved by the baseline $\setofFormulas_{1000}$ 
 versus \toolname on Z3 and CVC4's respective source codes.}
\label{coverage}
\end{figure*}
}

{\parindent0pt\paragraph{\textbf{Code coverage of OpFuzz's mutations}}
Code coverage is a reference for the sufficiency of software testing. This experiment 
aims to answer whether the mutants generated by \toolname can achieve higher coverage 
than the seed formulas. We randomly sampled 1000 formulas ($\setofFormulas_{1000}$) 
from all formulas that we used for stress-testing SMT solvers. 
We instantiated \toolname with $n = 300$, run \toolname on the seeds $\setofFormulas_{1000}$ 
and then measure the cumulative line/function/branch coverage over all formulas 
and runs.\footnote{This makes a total of 300k runs.} For all coverage measurements, we used 
Gcov\footnote{https://gcc.gnu.org/onlinedocs/gcc/Gcov.html} from the GCC suite.
}

The results show that \toolname 
increases the code coverage upon $\setofFormulas_{1000}$ (Fig.~\ref{coverage}). Z3 
and CVC4 have over 436K LoC and 238k LoC respectively, so that 0.1\% improvement 
already translate to hundreds of additionally covered lines. However, although 
noticeable, the coverage increments are not significant ($\leq$ 0.5\%). A partial 
explanation is that decision procedures of Z3 and CVC4 are highly recursive.    
This leads to many calls of the same functions with different arguments. 
Hence the difference in line/function/branch coverage achieved by different formulas of 
the same theory, may not be as significant. This experiment also provides further 
evidence that standard coverage metrics (\eg, statement and branch coverages), although 
useful, are insufficient for measuring the thoroughness of testing. 


\begin{figure}[t]
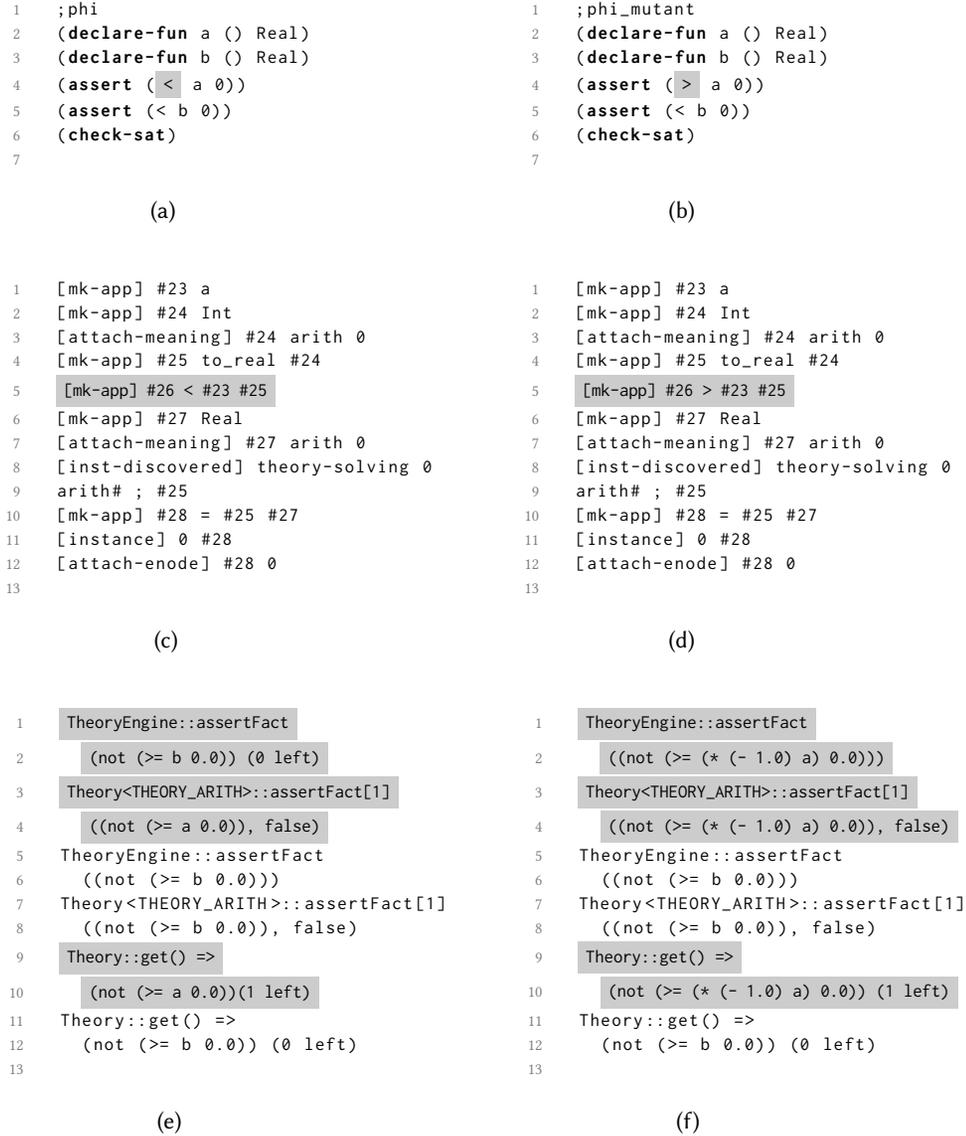

  \begin{subfigure}{0.49\textwidth}
  \centering
  \begin{lstlisting}[basicstyle=\scriptsize\ttfamily]
  ;phi 
  (declare-fun a () Real)
  (declare-fun b () Real)
  (assert (<@\colorbox{gray!40}{\texttt{<}}@> a 0))
  (assert (< b 0))
  (check-sat)
  \end{lstlisting}
  \caption{\hspace*{2.5cm} \label{trace:a}}
  \end{subfigure}
  \begin{subfigure}{0.49\textwidth}
  \centering
  \begin{lstlisting}[basicstyle=\scriptsize\ttfamily]
  ;phi_mutant
  (declare-fun a () Real)
  (declare-fun b () Real)
  (assert (<@\colorbox{gray!40}{\texttt{>}}@> a 0))
  (assert (< b 0))
  (check-sat)
  \end{lstlisting}
  \caption{\hspace*{2.5cm} \label{trace:b}}
  \end{subfigure}
  \vspace*{0.5cm}
  
  \begin{subfigure}{0.49\textwidth}
  \centering
  \begin{lstlisting}[basicstyle=\scriptsize\ttfamily]
  [mk-app] #23 a
  [mk-app] #24 Int
  [attach-meaning] #24 arith 0
  [mk-app] #25 to_real #24
  <@\colorbox{gray!40}{\texttt{[mk-app] \#26 < \#23 \#25}}@>
  [mk-app] #27 Real
  [attach-meaning] #27 arith 0
  [inst-discovered] theory-solving 0 
  arith# ; #25
  [mk-app] #28 = #25 #27
  [instance] 0 #28
  [attach-enode] #28 0
  \end{lstlisting}
  \caption{\hspace*{2.5cm}\label{trace:c}}
  \end{subfigure}
  \begin{subfigure}{0.49\textwidth}
  \centering
  \begin{lstlisting}[basicstyle=\scriptsize\ttfamily]
  [mk-app] #23 a
  [mk-app] #24 Int
  [attach-meaning] #24 arith 0
  [mk-app] #25 to_real #24
  <@\colorbox{gray!40}{\texttt{[mk-app] \#26 > \#23 \#25}}@>
  [mk-app] #27 Real
  [attach-meaning] #27 arith 0
  [inst-discovered] theory-solving 0 
  arith# ; #25
  [mk-app] #28 = #25 #27
  [instance] 0 #28
  [attach-enode] #28 0
  \end{lstlisting}
  \caption{\hspace*{2.5cm} \label{trace:d}}
  \end{subfigure}
  \vspace*{0.5cm}
  
  \begin{subfigure}{0.49\textwidth}
  \centering
  \begin{lstlisting}[basicstyle=\scriptsize\ttfamily]
  <@\colorbox{gray!40}{\texttt{TheoryEngine::assertFact}}@>
    <@\colorbox{gray!40}{\texttt{(not (>= b 0.0)) (0 left)}}@>
  <@\colorbox{gray!40}{\texttt{Theory<THEORY\_ARITH>::assertFact[1]}}@>
    <@\colorbox{gray!40}{\texttt{((not (>= a 0.0)), false)}}@>
  TheoryEngine::assertFact
    ((not (>= b 0.0)))
  Theory<THEORY_ARITH>::assertFact[1]
    ((not (>= b 0.0)), false)
  <@\colorbox{gray!40}{\texttt{Theory::get() =>}}@>
    <@\colorbox{gray!40}{\texttt{(not (>= a 0.0))(1 left)}}@>
  Theory::get() => 
    (not (>= b 0.0)) (0 left)
  \end{lstlisting}
  \caption{\hspace*{2.5cm}\label{trace:e}}
  \end{subfigure}  
  \begin{subfigure}{0.49\textwidth}
  \centering
  \begin{lstlisting}[basicstyle=\scriptsize\ttfamily]
  <@\colorbox{gray!40}{\texttt{TheoryEngine::assertFact}}@>
    <@\colorbox{gray!40}{\texttt{((not (>= (* (- 1.0) a) 0.0)))}}@>
  <@\colorbox{gray!40}{\texttt{Theory<THEORY\_ARITH>::assertFact[1]}}@>
    <@\colorbox{gray!40}{\texttt{((not (>= (* (- 1.0) a) 0.0)), false)}}@>
  TheoryEngine::assertFact
    ((not (>= b 0.0)))
  Theory<THEORY_ARITH>::assertFact[1]
    ((not (>= b 0.0)), false)
  <@\colorbox{gray!40}{\texttt{Theory::get() =>}}@>
    <@\colorbox{gray!40}{\texttt{(not (>= (* (- 1.0) a) 0.0)) (1 left)}}@>
  Theory::get() => 
    (not (>= b 0.0)) (0 left)
  \end{lstlisting}
  \caption{\hspace*{2.5cm}\label{trace:f}}
  \end{subfigure}
  \caption{Left column: (a) seed formula $\varphi$ (b) Z3 trace snippet of $\varphi$ 
  and (c) CVC4 trace snippet of $\varphi$. Right column: (d) type-aware operator mutant $\varphi_{\mathit{mutant}}$ 
  (e) Z3 trace snippet of $\varphi_{\mathit{mutant}}$ (f) CVC4 trace snippet of $\varphi_{\mathit{mutant}}$ 
  of $\varphi$. Differences in the execution traces snippets are shaded.
  \label{fig:formula-z3trace-cvc4trace}}
  \end{figure}

\begin{figure*}[t]
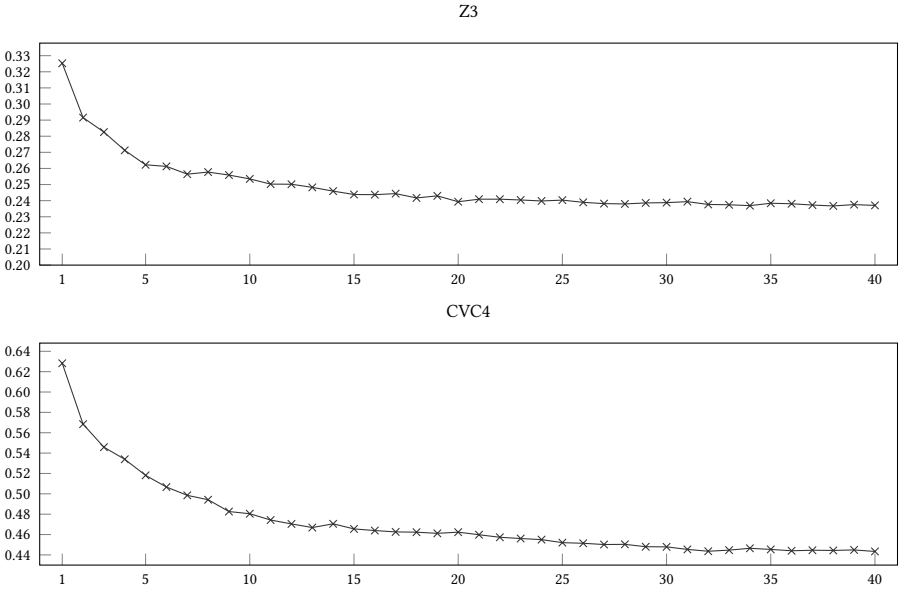

\zthreetrace

\cvcfourtrace
\caption{Average similarity of consecutively generated mutants (y-axis) per mutation step (x-axis).}
\label{trace}
\end{figure*}

{\parindent0pt\paragraph{\textbf{Execution trace comparison}}
Since code coverage could not thoroughly explain the effectiveness of \toolname, 
we also examine the internals of the solvers by investigating the similarity of 
their execution traces upon type-aware operator mutations. \emph{What is the relative 
similarity of the execution traces with respect to the seed?}  In the following experiment, 
we approach this question. In Z3 and CVC4, we can obtain an execution trace by 
setting the flags \texttt{\small TRACE=True} and \texttt{\small --trace-theory} respectively. 
Before describing this experiment, we first show the format of Z3's and CVC4's respective
traces via an example. Consider formula $\varphi$ and its type-aware operator mutation 
$\varphi_\mathit{mutant}$ (see Fig.~\ref{trace:a} and \ref{trace:d}).
Fig.~\ref{trace:c} and \ref{trace:e} shows Z3's and CVC4's traces on solving $\varphi$ respectively, 
Fig.~\ref{trace:d} and \ref{trace:f} show Z3's and CVC4's traces on solving $\varphi_\mathit{mutant}$ respectively. 
}

Having obtained an intuition of the execution traces, we now get to the actual 
experiment. Our aim is to measure the relative change in the execution traces      
of Z3 and CVC4. We therefore perform 40 mutation steps for every formula in $\setofFormulas_{1000}$ 
and record the execution trace triggered in each step. To quantify the similarity 
of two traces $t_1$ and $t_2$, we compute a metric $sim(t_1, t_2)$ with 
$$sim(t_1, t_2) =  \frac{2 \cdot |LCS(t_1, t_2)|}{\#lines(t_1) + \#lines(t_2)}$$ 
where $LCS(t_1, t_2)$ corresponds to the longest common subsequence of    
$t_1$ and $t_2$; $\#lines(t_1)$ and $\#lines(t_2)$ are the number of lines in 
$t_1$ and $t_2$ respectively. As an example, re-consider Fig.~\ref{fig:formula-z3trace-cvc4trace}. 
The differing lines of $\varphi$'s trace and $\varphi_{\mathit{mutant}}$'s trace are shaded.
$\varphi$'s Z3 trace and $\varphi_{\mathit{mutant}}$'s Z3 trace match in $10$ out of $11$   
lines and therefore their similarity score is $\frac{10}{11}$. For the trace pair 
of CVC4, the number of longest common subsquence is of length 3 and hence 
the similarity of Z3's trace is $\frac{1}{2}$. To compute the longest common subsequence,   
we used the \texttt{difflib}\footnote{https://docs.python.org/3/library/difflib.html} package from python's standard library.
Note that type-aware operator mutation may rename the AST node identifiers of Z3's trace.
Here, we under-approximate the similarity of Z3 traces by considering the 
identifier renaming as the change of the trace.

In our experiment, we fix the trace of the original formula to be $t_1$, and $t_2$  
corresponds to the trace triggered of the mutant. Fig.~\ref{trace} shows the 
similarity of the corresponding mutation step averaged over all formulas in  
$\setofFormulas_{1000}$. The results of Z3 and CVC4 consistently show that along 
with a gradual mutation step increase, the similarity between the traces triggered 
by the mutant and the original formula gradually decreases. The result indicates 
that \toolname can generate diverse test cases that trigger different execution 
traces via type-aware operator mutation.  

{\parindent0pt \paragraph{\textbf{Takeaways}}
We designed three quantitative evaluations to measure and gain an intuition 
about the effectiveness of \toolname. First, we observe that  \toolname can find 
a significant number of bugs in various logics, solver configurations, most of 
which are in default mode. Second, to understand why \toolname can find so many 
bugs, we designed a coverage evaluation. The evaluation result shows that \toolname 
can increase coverage, but the increment is minor. As the coverage evaluation 
did not answer why \toolname is effective, we further designed the third 
evaluation investigating the similarity of execution traces. The trace evaluation 
shows that \toolname can gradually change the execution traces of the solvers, 
which partially explains the effectiveness of \toolname.
}

%% file: sections/bug-samples.tex
\section{In-depth Bug Analysis}
\label{sec:bug-study}
This section presents an in-depth study on \toolname's bug findings 
in which we (1) quantify the fixing efforts for Z3's and CVC4's developers
(2) identify weak components in Z3 and CVC4 and (3) examine the file sizes of bug 
triggering SMT formulas. We summarize the insights gained and then    
present selected bug samples, examine their root causes, and the developer's fixes.    

{\parindent0pt\subsection{\textbf{Quantitative Analysis}}
We collected all GitHub bug reports that we have filed in our extensive evaluation   
of \toolname. This data serves as the basis for our analysis. We guide our analysis 
with three research questions.  
\paragraph{\textbf{RQ1: How much effort had the developers taken with fixing the bugs found by \toolname?}}
To approach this question, we consider two metrics: the files affected by a bug fix
and the number of lines of code (LoC) for fixing. If a bug causes many 
lines of code and/or file changes, this may indicate a high fixing effort necessary    
by the developers. To examine the LoC and file changes for the bugs found by \toolname,
we collected \zcommits bug fixing commits in Z3 and \cvccommits bug fixing commits in CVC4. 
We solely considered commits that could be one-to-one matched to their bug reports 
\ie, the commit log mentions the issue id of our bug reports and no other issue ids.   
}
\begin{figure}[t]
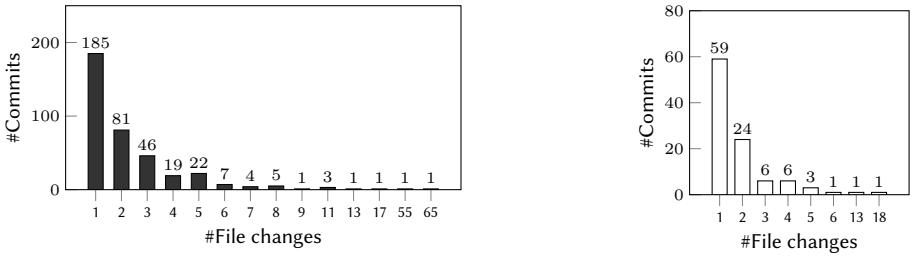

   \begin{subfigure}{0.60\textwidth}
    \commitsFilesZ
   \end{subfigure}
   \begin{subfigure}{0.39\textwidth}
    \commitsFilesCVC
   \end{subfigure}

    \caption{Distributions of the file changes for a single bug-fixing commit in 
             Z3 (left) and CVC4 (right).}
    \label{fig:filescommits}
\end{figure}
%
Fig.~\ref{fig:filescommits} shows the distributions of file changes for 
bug-fixing commits in Z3 (left) and CVC4 (right). 
We observe that, in both solvers, most bug-fixing commits change less than five files, 
and the single file fixes are the majority. However, a few commits affected many 
files. We have manually examined the right tail of the distribution. We specifically 
present the top-$2$ file changing commits in both Z3 and CVC4 individually to 
demonstrate exemplary reasons for major changes in Z3 and CVC4. 
We begin with Z3. The highest-ranked bug-fixing commit in Z3 triggered 65 changes.
The main part of this fix was in \texttt{\small "smt/theory\_bv.cpp"} which is the implementation of 
bit-vector logic and also serves as the low-level implementation for floating-point 
logic. The developers' fix resulted in many function name updates and added checkpoints 
and additional $64$ file changes. Another bug-fixing commit in Z3 that affected 
$55$ files is a crash. It was caused by an issue in Z3's abstract syntax tree. 
The core issue addressed by this fix was in \texttt{\small "ast/rewriter/rewriter\_def.h"} 
and \texttt{\small "ast/rewriter/th\_rewriter.cpp"}. 
Reorganizations of the assertion 
checks triggered additional $54$ file changes. 
In CVC4, the top-2 fixes with most file changes ($18$ and $13$) are both caused by the crash bugs affecting  
string operators. The first bug is due to the unsound variable elimination that triggered the assertion violation. 
The fix refactored the variable elimination with 295 LoC changed.
For fixing the second bug, the developers added support for the regex 
operators \texttt{\small re.loop} and \texttt{\small re.\^{}} that were recently added to the   
theory of strings.  As an intermediate conclusion 
we observe: Although a relatively high numbers of file changes indicate extensive 
revisions in the SMT solvers Z3 and CVC4, their root cause are often rather simple 
fixes such as updating  function names, adding assertions, \etc We therefore 
also investigate the LoC changes for each bug fixing commit. Many simple fixes, on 
the other hand, exhibit subtle missed corner cases. 


\begin{figure}[t]
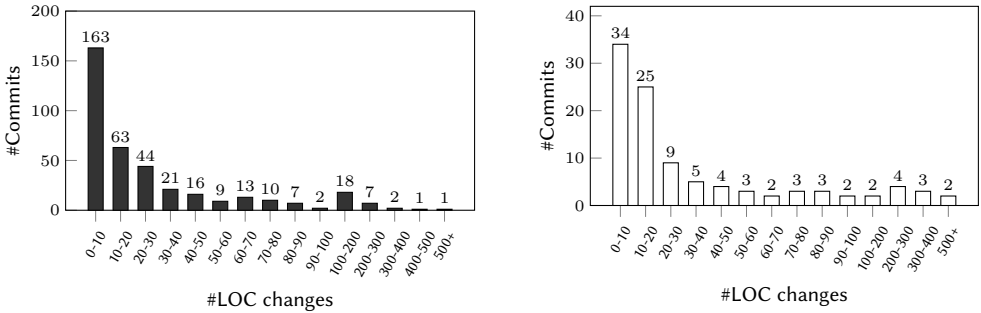

        \begin{subfigure}{0.505\textwidth}
            \commitsLoCZ
        \end{subfigure}
        \begin{subfigure}{0.485\textwidth}
        \commitsLoCCVC
        \end{subfigure}
    \vspace{-0.2cm}
    \caption{Distributions of the LoC changed in one commit in Z3 (left) and CVC4 (right). \label{fig:LOCcommits}}
    \vspace{-0.4cm}
\end{figure}
Fig.~\ref{fig:LOCcommits} presents the distributions of the LoC changes for 
each bug-fixing commit. For both Z3 and CVC4, we observe that most commits have less 
than 100 LoC changes and many bugs fixes only involve a 0-10 LoC change. 
We have manually inspected all 0-10 LoC fixes and observed the majority       
of them are subtle corner cases. Again we examine top-2 commits 
for each solver. In Z3 these have 572 and 332 LoC changes respectively. 
The 572 LoC change commit is a fix for a soundness bug in string logic. 
It leads to an extensive change in the rewriter of the sequential solver.
The 481 LoC changes commit is a fix for a soundness bug in non-linear 
arithmetic logic. The fix was systematically revamping the decoupling of monomials
in non-linear arithmetic logic. 
For CVC4, the top-2 commits have 1162 and 588 LoC changes respectively. 
The 1162 LoC changes in CVC4 
commit fixes a crash bug by systematically removing the instantiation propagator 
infrastructure of CVC4. The developer commented that they will redesign this 
infrastructure in the future. The bugfix with 588 LoC changes is fixing a soundness 
bug which is labeled as "major" in CVC4's issue tracker. 
The bug is due to a buggy ad-hoc rewriter that was incorporated into CVC4's  
extended quantifier rewriting module. The fix deleted the previous buggy rewriting steps 
and re-implemented an alternative rewriter. Compared to the analysis of file changes, 
commits with high LoC have a stronger correlation with interesting and systematic fixes in the SMT solvers. 
On average, the bugs found by \toolname lead to \zlocchanges and \cvclocchanges 
LOC changes for each commit in Z3 and CVC4 respectively. 

\begin{figure}[t]
    \centering
    \begin{subfigure}{0.49\textwidth}
        \centering
        \filesCommitsZ
        \caption{}
        \label{fig:rankcommitsZ3}
    \end{subfigure}
    \begin{subfigure}{0.49\textwidth}
        \centering
        \filesChangesZ
        \caption{}
        \label{fig:rankchangesZ3}
    \end{subfigure}

    \vspace{0.6cm}

    \begin{subfigure}{0.49\textwidth}
        \centering
        \filesCommitsCVC
        \caption{}
        \label{fig:rankcommitsCVC4}
    \end{subfigure}
    \begin{subfigure}{0.49\textwidth}
        \centering
        \filesChangesCVC
        \caption{}
        \label{fig:rankchangesCVC4}
    \end{subfigure}
    \caption{Top-10  (a) files affected by bug fixing commits in Z3. (b) LoC 
    changes per file in Z3 (c) files affected by bug fixing commits in CVC4.
    (d) LoC changes per file in CVC4.\label{fig:ranking} 
    }
\end{figure}
{\parindent0pt\paragraph{\textbf{RQ2: Which parts/files of Z3's and CVC4's codebases are most affected by the fixes?}}
In this research question, we investigate the influence of \toolname's bug findings      
on the respective codebases of Z3 and CVC4. For this purpose, we use two metrics.  
First, the number of bug-fixing commits that changed a specific file $f$ 
in either Z3's or CVC4's codebase, \ie, in how many bug-fixing commits 
file $g$ was included. The second metric is the cumulative number of LoC 
changes for a file $f$ caused by fixes in either Z3's or CVC4's codebase.        
For each file $f$ we add up additions and deletions based on GitHub's changeset. 
}

In general, there are 103 files in CVC4 and 348 files in Z3 are affected by the fixes 
of our bugs. 
Fig.~\ref{fig:ranking} shows a top-$10$ ranking of files in Z3's (top row) 
and CVC4's codebase (bottom row) with respect to the two metrics. 
We observe that in both Fig.~\ref{fig:rankcommitsZ3} and Fig.~\ref{fig:rankchangesZ3}, most files belong to 
the \texttt{\small "smt"} directory which contains the core implementations of Z3. 
Strikingly, the files \texttt{\small "smt/theory\_seq.cpp"} (Z3's sequence and string solvers), 
\texttt{\small "smt/theory\_arith\_nl.h"} (Z3's non-linear arithmetic solver)
and \texttt{\small "smt/theory\_lra.cpp"} (Z3's linear arithmetic solver)
are ranked in the top-$6$ in both \#Commits and \#LoC changes rankings.
This suggests that many of \toolname bug findings lead to the fixes in the core components of Z3.
Besides files from the \texttt{\small "smt"} directory, the remaining files    
are mostly part of Z3's \texttt{\small"tactic"} and \texttt{\small "ast"} directories.
These contain the implementations of solver front-end and Z3's solving tactics.    
Note, several formulas rewriters related files such as files \texttt{\small "ast/rewriter/seq\_rewriter.cpp"}, 
\texttt{\small "ast/rewriter/rewriter\_def.h"} and \texttt{\small "tactic/ufbv/ufbv\_rewriter.cpp"} 
are also highly ranked in the top-10 files affected by the fixes.

We now turn our attention to CVC4. Consider the bottom row of Fig.~\ref{fig:ranking} that presents file and   
LoC rankings in CVC4. The files that are related to quantifiers (under the path \texttt{\small "theory/quantifiers"})
are the majority in both rankings.
Besides, the files \texttt{\small "nonlinear\_extension.cpp"}, \texttt{\small "theory\_strings.cpp"}
and \texttt{\small "quantifiers\_rewriter.cpp"}
are listed in both rankings. The file \texttt{\small "nonlinear\_extension.cpp"} 
was the implementation of non-linear arithmetic solver, and a recent pull request 
moved the core of the non-linear arithmetic solver elsewhere. The file 
\texttt{\small "quantifiers\_rewriter.cpp"} contains the implementations of 
quantifier rewriters that caused soundness bugs, as RQ1 revealed. 
The file \texttt{\small "theory\_strings.cpp"} contains the decision procedures 
for string logic in CVC4.  Moreover the model generator of non-linear arithmetic 
(\texttt{\small "nl\_model.cpp"}) and the pre-processor (\texttt{\small "int\_to\_bv.cpp"}, 
\texttt{\small "unconstrained\_simplifier.cpp"}) are also heavily influenced by bug fixes.

\begin{figure}[t]
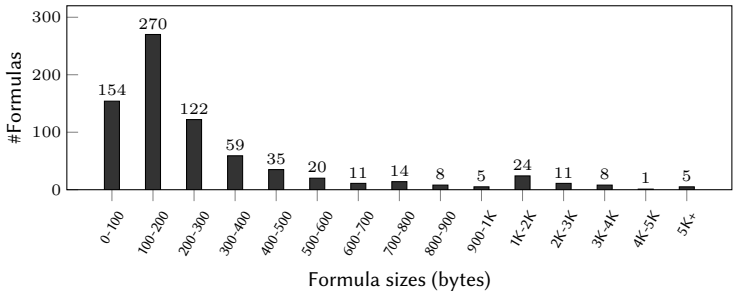

    \formulasize
    \vspace{-0.2cm}
    \caption{File-size distribution of reduced bug-triggering formulas.}
    \label{fig:size}
\vspace{-0.5cm}
\end{figure}
{\parindent0pt\paragraph{\textbf{RQ3: What is the file-size distribution of the bug-triggering formulas?}}
In this research question, we investigate the file-size distribution of reduced bug-triggering 
formulas. We collected the bug-triggering formulas from all confirmed and fixed Z3 and CVC4 bug reports we filed. 
Fig.~\ref{fig:size} presents the distribution of bug-triggering formulas 
collectively for Z3 and CVC4. According to Fig.~\ref{fig:size}, most formulas 
have less than 600 bytes, while the range of 100-200 bytes has the highest formula count. 
Among all bugs-triggering formula we reported, there are three formulas to have more than 
10,000 bytes \ie, 23,770 bytes, 19,473 bytes, and 10,562 bytes. All of these are 
in bit-vector logic. The formula with 23,770 bytes and 10 562 bytes triggered 
an invalid model bug and a crash bug respectively in Z3, and both of them take 
the developers a half month to fix.
The formula with 19 473 bytes triggered a crash bug in CVC4.
}

The top-3 smallest bug-triggering formulas have 21, 30, and 34 bytes 
respectively. The 21-byte formula is an invalid formula that triggers a crash bug 
in Z3.
The 30 bytes and 34 bytes formulas are a crash-triggering formula for Z3 and CVC4 respectively, 
both point to the corner cases. 
These three bugs were all fixed promptly, i.e., in less than one day, which is significantly
faster than the bugs triggered by the top-3 largest formulas. In general, the average 
size of the bug-triggering formulas reported by us is 426 bytes, which is usually 
sufficiently small for the developers.   

\subsection{Insights}
{\parindent0pt\paragraph{\textbf{Insight 1: \toolname's Bugs Are of High Quality}}
RQ1 and RQ2 have shown that \toolname's bug findings have not only led to non-trivial file 
and LoC changes in both CVC4 and Z3, but also motivated the developers to reorganize and 
redesign some parts of the solvers. Systematic infrastructure changes such as the decoupling    
of the monomial instantiation propagator show this. Furthermore, \toolname's 
bugs affected core implementations of the SMT solvers Z3 and CVC4. As RQ2 presented, 
the \texttt{\small "smt"} directory in Z3 and \texttt{\small "theory"} directory in CVC4, 
solvers are among the most affected. Besides, the bugs also affected various 
pre-processors and rewriters components. Third, the bug-triggering formulas that \toolname 
could be reduced to reasonable sizes (cf. RQ3). 
}

{\parindent0pt\paragraph{\textbf{Insight 2: weak components in Z3 and CVC4}}
From the rankings in RQ2, we identify several "weak" components in Z3 and CVC4.
First, in both Z3 and CVC4, source files for the non-linear arithmetic solvers 
rank high. This indicates: decision procedures for non-linear arithmetic are among the  
weak components in SMT solvers. Apart from these, rewriters are weak 
components as well. Z3's \texttt{\small "rewriter\_def.h"}, \texttt{\small "ufbv\_rewriter.cpp"} and 
\texttt{\small "seq\_rewriter.cpp"} are among the top-10 in LoC changes.  
In CVC4, the quantifier rewriter \texttt{\small "quantifiers\_rewriter.cpp"} is ranked high 
(5th and 2nd in Fig.~\ref{fig:rankcommitsCVC4} and Fig.~\ref{fig:rankchangesCVC4} respectively).
In Z3, we identified the tactics to be a weak component. Among the filed bug reports, 
there are \ztactics including reports related to tactics. 
In Fig.~\ref{fig:ranking}, these are \texttt{\small "purify\_arith\_tactic.cpp"} 
and \texttt{\small "dom\_simplify\_tactic.cpp"} which are ranked 9th or 10th in both 
Fig.~\ref{fig:rankcommitsZ3} and Fig.~\ref{fig:rankchangesZ3}.
}
 



{\parindent0pt\paragraph{\textbf{Insight 3: bugs found by \toolname can usually be reduced to small-sized 
formulas but bug reduction can be challenging}}
As we have observed (c.f. Fig.~\ref{fig:size}), 90\% of all bugs found by \toolname
are triggered by formulas of less than 600 bytes. Small-sized formulas facilitate   
the bug fixing efforts significantly. As we observed in RQ3, the  
top-$3$ largest formulas took the developers around half a month while the top-$3$ smallest 
formulas have been fixed very fast, usually within just a few hours. However, 
reducing SMT formulas to such small sizes can be challenging.    
ddSMT~\cite{NiemetzBiere-SMT13} is the only existing specialized SMT formula 
reducer for that purpose which does however not fully support the SMT-LIB 2.6 standard 
and formulas in string logic. We therefore preferred C-Reduce, a C code reducer 
to reduce SMT formulas. While creduce worked well in practice, bug 
reduction is often challenging especially if the time for solving the formula is high. 
}

\subsection{Assorted Bug Samples}
This subsection details multiple bug samples from our extensive bug hunting campaign  
of the SMT solvers Z3 and CVC4 and inspects the root causes. The bugs shown 
are reduced by C-Reduce, since the unreduced formulas are too large to be presented. 


\begin{figure}
%
\begin{subfigure}{0.45\textwidth}
\centering
\begin{lstlisting}[basicstyle=\scriptsize\ttfamily]
(declare-const a (_ BitVec 8)) 
(declare-const b (_ BitVec 8)) 
(declare-const c (_ BitVec 8)) 
(assert (=  (bvxnor a b c) 
    (bvxnor (bvxnor a b) c)))
(check-sat) 
\end{lstlisting}
\vspace{-0.1cm} 
\caption{Soundness bug in Z3 caused by a logic in the handling of the ternary 
\texttt{bvxnor} operator.\\[0.1cm] 
\scriptsize \url{https://github.com/Z3Prover/z3/issues/2832}
\label{sample:a}}
\vspace{0.3cm}
\end{subfigure}
\hspace{0.8cm}
\begin{subfigure}{0.45\textwidth}
\centering
\begin{lstlisting}[basicstyle=\scriptsize\ttfamily]
(set-logic ALL)
(declare-fun x () Real)
(assert (< x 0))
(assert (not (= 
    (/ (sqrt x) (sqrt x)) x)))
(check-sat)
\end{lstlisting}
\vspace{-0.1cm} 
\caption{Soundness bug in CVC4 caused by an inadmissble reduction of the square 
root operator.
\\[0.1cm] \scriptsize \url{https://github.com/CVC4/CVC4/issues/3475} 
\label{sample:b}}
\vspace{0.3cm}
\end{subfigure}
\begin{subfigure}{0.45\textwidth}
\centering
\begin{lstlisting}[basicstyle=\scriptsize\ttfamily]
(declare-fun a () Int)
(declare-fun b (Int) Bool)
(assert (b 0)) (push)
(assert (distinct true
    (= a 0) (not (b 0))))
(check-sat)
\end{lstlisting}
\vspace{-0.1cm} 
\caption{Soundness bug in Z3 in the boolean rewriter handling the distinct       
operator. \\[0.1cm]
\scriptsize \url{https://github.com/Z3Prover/z3/issues/2830} 
\label{sample:c}}
\vspace{0.3cm}
\end{subfigure}
\hspace{0.8cm}
%
\begin{subfigure}{0.45\textwidth}
\centering
\begin{lstlisting}[basicstyle=\scriptsize\ttfamily]
(set-logic QF_AUFBVLIA)
(declare-fun a () Int)
(declare-fun b (Int) Int)
(assert (distinct (b a) 
    (b (b a))))
(check-sat)
\end{lstlisting}
\vspace{-0.1cm} 
    \caption{Soundness bug in CVC4 due to a variable re-use in a simplification. 
    \\[0.1cm] \scriptsize \url{https://github.com/CVC4/CVC4/issues/4469}   
\label{sample:d}}
\vspace{0.3cm}
\end{subfigure}

\begin{subfigure}{0.45\textwidth}
\centering
\begin{lstlisting}[basicstyle=\scriptsize\ttfamily]
(declare-fun a () String)
(declare-fun b () Int)
(assert (> b 0))
(assert (= (int.to.str b) 
    (str.++ "0" a)))
(check-sat)
\end{lstlisting}
\vspace{-0.1cm} 
\caption{Soundness bug in Z3 due to a missing axiom in the integer to string 
         conversion function. \\[0.1cm]
\scriptsize \url{https://github.com/Z3Prover/z3/issues/2721}
\label{sample:e}}
\vspace{0.3cm}
\end{subfigure}
\hspace{0.8cm}
\begin{subfigure}{0.45\textwidth}
\centering
\begin{lstlisting}[basicstyle=\scriptsize\ttfamily]
(declare-fun x () String)
(declare-fun y () String)
(assert (= (str.indexof x y 1) 
    (str.len x)))
(assert (str.contains x y))
(check-sat)
\end{lstlisting}
\vspace{-0.1cm} 
\caption{Soundness bug in CVC4 due to an invalid indexof range lemma. 
\\[0.1cm]  
\scriptsize \url{https://github.com/CVC4/CVC4/issues/3497}
\label{sample:f}}
\vspace{0.3cm}
\end{subfigure}

\begin{subfigure}{0.45\textwidth}
\begin{lstlisting}[basicstyle=\scriptsize\ttfamily]
(declare-fun a () Real)
(assert (forall ((b Real)) 
    (= (= a b) (= b 0))))
(check-sat-using qe)
\end{lstlisting}
\vspace{-0.1cm} 
\caption{Longstanding soundness bug in Z3's 
qe tactic (since version \texttt{4.8.5}). \\[0.1cm]
\scriptsize \url{https://github.com/Z3Prover/z3/issues/4175}
\label{sample:g}}
\vspace{0.3cm}
\end{subfigure}
\hspace{0.8cm}
\begin{subfigure}{0.45\textwidth}
\centering
\begin{lstlisting}[basicstyle=\scriptsize\ttfamily]
(declare-fun a () Real)
(assert (= (* 4 a a) 9))
(check-sat)
(get-model)
\end{lstlisting}
\vspace{-0.1cm} 
\caption{Invalid model bug in CVC4 caused by an incorrect implementation of the  
         square root.
    \\[0.1cm] \scriptsize \url{https://github.com/CVC4/CVC4/issues/3719}
\label{sample:h}}
\vspace{0.3cm}
\end{subfigure}

\begin{subfigure}{0.45\textwidth}
\centering
\vspace{0.9cm} 
\begin{lstlisting}[basicstyle=\scriptsize\ttfamily]
(declare-fun a () Int)
(declare-fun b () Real)
(declare-fun c () Real)
(assert (> a 0))
(assert (= (* (/ b b) c) 2.0))
(check-sat)
(check-sat)
(get-model)
\end{lstlisting}
\caption{Invalid model bug in Z3. For the second \texttt{\small check-sat} query, Z3 returns an invalid model.
\\[0.1cm] \scriptsize \url{https://github.com/Z3Prover/z3/issues/3118}
\label{sample:i}}
\end{subfigure}
\hspace{0.8cm}
\begin{subfigure}{0.45\textwidth}
\centering
\begin{lstlisting}[basicstyle=\scriptsize\ttfamily]
(declare-fun d () Int)
(declare-fun b () (Set Int))
(declare-fun c () (Set Int))
(declare-fun e () (Set Int))
(assert (subset b e))
(assert (= (card b) d))
(assert (= (card c) 0 (mod 0 d)))
(assert (> (card (setminus e 
    (intersection (intersection e b) 
    (setminus e c)))) 0))
(check-sat)
\end{lstlisting}
\vspace{-0.1cm}
\caption{Soundness bug in CVC4's set logic caused by an incorrectly implemented 
    cardinality rule. 
    \\[0.1cm] \scriptsize \url{https://github.com/CVC4/CVC4/issues/4391}
    \label{sample:j}
}
\end{subfigure}
\caption{Selected bug samples in Z3 and CVC4. \label{fig:bug-array-samples}}
\end{figure}
    %
    {\parindent0pt\paragraph{Fig.~\ref{sample:a}} shows a soundness bug in Z3's bit-vector logic.      
    The formula is clearly unsatisfiable as the nested \texttt{\small bvxnor} expression 
    equals the unnested \texttt{\small bvxnor} expression. However, Z3 reports 
    \texttt{\small unsat}  on it, which is incorrect.  
    The root cause for this bug  is an incorrect handling of the ternary 
    \texttt{\small bvxnor} in Z3's bitvector rewriter "\texttt{\small bv\_rewriter.cpp}".    
    The \texttt{\small bvxnor} was implemented as the negation of the \texttt{\small bvxor} operator.    
    This is correct in the binary case, however incorrect for the n-ary case.  
    To see this consider, \eg, $\texttt{\small (bvxnor a b c) = (bvxnor (bvxnor a b) c) = true} 
    \neq \texttt{\small (not (bvxor a b c)) = false}$ for $\texttt{\small a = b = c = true}$.
    In the fix, Z3's main developer recursively reduces n-ary case \texttt{\small bvxnor} expression 
    to the binary case. The fix lead to a 17 LoC change in \texttt{\small ast/rewriter/bv\_rewriter.cpp}.}
    %
    {\parindent0pt\paragraph{Fig.~\ref{sample:b}} shows a soundness bug in the implementation    
    of the symbolic square root in CVC4. The formula can be satisfied by assigning 
    an arbitrary negative real to variable \texttt{\small x}. CVC4 incorrectly 
    reported \texttt{\small unsat} on this formula. The root cause for this bug   
    is an inadmissible reduction of the square root expression $\texttt{\small (sqrt x)}$     
    to "choice real $\texttt{\small y}$ s.t. $\texttt{\small x} = \texttt{\small y} 
    \cdot \texttt{\small y}$".  For negative $\texttt{\small x}$, there is  
    no $\texttt{\small y}$  to satisfy the equation. However, square roots 
    of negative numbers are permitted by the SMT-LIB standard.  CVC4's developers fixed this bug by      
    interpreting square roots of negative numbers as an undefined value that  
    can be chosen arbitrarily. For the formula in Fig.~\ref{sample:b}, the term  
    $\texttt{\small (/ (sqrt x) (sqrt x))}$ can be arbitrarily chosen,  as the  
    second assert demands $\texttt{x}$ to be negative. Therefore,  the   
    formula in Fig.~\ref{sample:b} is satisfiable. The bug-fixing pull request 
    was labeled as "major" which reveals that this issue was of high importance 
    to the CVC4 developers. The fix lead to a 126 LoC change on 5 files. }

    %
    {\parindent0pt\paragraph{Fig.~\ref{sample:c}} is a soundness bugs in Z3. 
    Although the second assert is unsatisfiable (as \texttt{\small true} cannot be distinct 
    with \texttt{\small (not (b 0))}, Z3 reported \texttt{\small sat} on this formula.
    The bug is caused by a logic error in a loop condition of a rewriting rule   
    for the \texttt{\small distinct} operator. An incorrect index condition accidentally    
    skips the last argument in an n-ary distinct. The \texttt{\small push} command 
    is necessary for triggering the bug, as it actives 
    the rewriter for \texttt{\small distinct}. The developer has fixed this bug  
    by correcting the index condition.  Hence, his fix consisted of only two character 
    deletes in \texttt{\small ast/rewriter/bool\_rewriter.cpp}.}
    
    \vspace{-3pt}
    {\parindent0pt\paragraph{Fig.~\ref{sample:d}} shows a  soundness bug in CVC4's logic of uninterpreted functions
    In default mode, CVC4 incorrectly reports \texttt{\small unsat} on this satisfiable formula.
    If we disable unconstrained simplification (\texttt{\small --no-unconstrained-simp}), 
    CVC4 correctly reports \texttt{\small unsat}. The bug is caused by an unsound 
    variable reuse. Our bug report got a "major" label by CVC4's developers and was promptly fixed.
    The core fix consisted of three LoC deletions in file the unconstrained simplifier implementation 
    "\texttt{\small preprocessing/passes/unconstrained\_simplifier.cpp}".}
    
    \vspace{-3pt}
    {\parindent0pt\paragraph{Fig.~\ref{sample:e}} depicts a soundness bug in Z3's QF\_SLIA logic.  
    The formulas is unsatisfiable, since if assertion \texttt{\small b > 0} holds, there   
    does not exist an \texttt{a} starting with \texttt{"0"}. However, Z3 reports \texttt{sat} 
    on this formula. The developers fixed this issue by adding an axiom to the 
    \texttt{\small smt/theory\_seq.cpp} adding two additional LoC to this file.}

    \vspace{-3pt}
    {\parindent0pt\paragraph{Fig.~\ref{sample:f}} shows a soundness in CVC4's string logic.     
    The intuition behind this formula is the following. The index of string \texttt{\small y} 
    in \texttt{\small x} after position 1 should be equal to the length of string \texttt{\small x}. 
    Furthermore \texttt{\small x} should contain \texttt{\small y}. The formula can be satisfied by 
    setting \texttt{\small y} to the empty string and \texttt{\small x} to a string of length $1$. 
    However, CVC4 incorrectly reports \texttt{\small unsat}. The root cause  
    was a logic error in  \texttt{\small theory/strings/theory\_strings.cpp} 
    The developer's fix changed three characters in \texttt{\small theory/strings/theory\_strings.cpp}. 
    The fix is labelled as "major".}

    \vspace{-3pt}
    %
    {\parindent0pt\paragraph{Fig.~\ref{sample:g}} presents a long-standing soundness bug in Z3's \texttt{\small qe} tactic.
    It affects z3 release from version \texttt{\small 4.8.5} to \texttt{\small 4.8.7}. 
    The \texttt{\small qe} tactic is an equisatisfiable transform for eliminating  
    quantifiers. Hence, the satisfiability should not be changed by using the 
    \texttt{\small qe} tactic.  The formula is satisfiable by assigning \texttt{\small a} to 
    \texttt{\small 0}, while Z3's \texttt{\small qe} tactic reports \texttt{\small unsat}. 
    The bug has been confirmed by Z3's developers but has not been fixed yet.}

    \vspace{-3pt}
    %
    {\parindent0pt\paragraph{Fig.~\ref{sample:h}} shows an invalid model bug in CVC4. CVC4 correctly 
    reports \texttt{\small sat} but generates the model ${\{a \mapsto \frac{-9}{2}\}}$ 
    which does not satisfy the formula. The bug is caused in CVC4's implementation of the square root.   
    A logic error assigns the result of the square root to be the square root's argument.  
    The fix is labeled as "major" by the developers, and promptly fixed only with a two LoC change 
    in file \texttt{\small theory/arith/nl\_model.cpp}.}  

    \vspace{-3pt}
    {\parindent0pt\paragraph{Fig \ref{sample:i}} shows an invalid model bug in Z3. 
    The \texttt{\small (check-sat)} command appears twice in the formula. 
    This means that Z3 is queried twice for solving. Z3 reports \texttt{\small unknown}
    for the first query and \texttt{\small sat} for the second.  
    In the second query, Z3 gives the following invalid model   
    {\small $ \{a \mapsto 0, b \mapsto 0.0, c \mapsto 16.0, \frac{0}{0} \mapsto \frac{1}{8}\}$}
    violating \texttt{\small (> a 0)}. The developers fixed this bug through 
    three LoC changes in file \texttt{\small solver/tactic2solver.cpp}.} 
    %

    \vspace{-3pt}
    {\parindent0pt\paragraph{Fig.~\ref{sample:j}} 
    presents a soundness bug in CVC4's set logic.  
    CVC4 returns \texttt{unsat} on this satisfiable formula. The root cause is  
    an incorrectly implemented set cardinality rule in the cardinality extension   
    of CVC4. CVC4's set solver uses lemmas to guess the equalities for terms by 
    identifying cycles of terms $e_1= \cdots = e_2 = \cdots = e_2$. CVC4 has  
    incorrectly assumed that these cycles are loops and in that case would conclude      
    $e_1= \cdots = e_2$. However, the cycles could have a lasso form which was 
    triggered by our formula. The developers fixed this issue, included the 
    formula to CVC4's regression test suite and marked the pull request to be 
    critical for CVC4's 1.8 release.}
    The fix was labeled as "major" and made 9 LoC changes to 
    \texttt{\small theory/sets/cardinality\_extension.cpp}.   

%% file: sections/related-work.tex
\section{Related work} \label{related-work}
\label{sec:related-work}
\vspace{10pt}
{\parindent0pt\paragraph{\textbf{Testing SMT solvers.}}This paper is not the first work on testing SMT solvers. Roughly ten years ago, 
the fuzzing tool FuzzSMT \citep{brummayer2009fuzzing} has been proposed, which is based         
on differential testing and targeted bit-vector logic. FuzzSMT 
uses a grammar for generating the SMT formula. FuzzSMT totally found $16$ 
solver defects in five solvers, however, none in Z3. BtorMBT~\cite{niemetz2017model}
is a testing tool for Boolector~\cite{BiereBrummayer-TACAS09}, an SMT 
solver for the bit-vector theory. BtorMBT tests Boolector by generating 
random valid API call sequences. However, BtorMBT did not find any bugs in a real setting. 
}

The efforts of the SMT-LIB initiative \citep{SMTLIB} have resulted in formalized 
SMT theories and common input/output file formats. In addition, the yearly solver 
competition SMT-COMP \citep{SMT-COMP} heavily penalized solvers with soundness issues.
Consequently, SMT solvers have robustified and finding bugs in SMT solvers became 
more difficult. Researchers have hence targeted the less mature logics such as 
the recently proposed theory of strings. Blotsky \etal \citep{blotsky-et-al-cav18} 
proposed StringFuzz which focuses on performance issues in string logic. 
StringFuzz generates test cases in two ways, one is mutating and transforming 
the benchmarks, another one is randomly generating formulas from a grammar. 
StringFuzz found $2$ performance bugs and $1$ implementation bug in z3str3. 
Bugariu and M\"uller~\citep{bugariu2020automatically} proposed a formula synthesizer 
for String formulas that are by construction satisfiable or unsatisfiable. 
They showed that their approach can detect many 
existing bugs in String solvers and they found $5$ new soundness/incorrect model 
bugs in z3 and z3str3. However, it remained an open question whether automated 
testing tools could find bugs in theories except the unicode string theory in Z3 and CVC4.
Recently, semantic fusion \citep{semantic-fusion} has been proposed which is an approach 
to stress-test  SMT solvers by fusing formula pairs that are by construction either 
satisfiable or unsatisfiable. \citeauthor{semantic-fusion}'s tool YinYang found $39$ bugs in Z3 and $9$ in  
CVC4. STORM \cite{mansur-etal-arxiv2020}, another recent mutation-based 
SMT solver testing approach, found $27$ bugs in Z3, however none in CVC4.   
Another related approach is BanditFuzz~\cite{scott-et-al-cav20}, a reinforcement 
learning-based fuzzer to detect SMT solver performance issues.  

Compared to previous work, type-aware operator mutation is the simplest, 
while it has also demonstrated to be the most effective technique for testing SMT solvers. 
Type-aware operator mutations show a promising direction for testing SMT solvers      
which can benefit the whole community. For example, \toolname can be used for the 
solver developers to stress-test new features conveniently.   

{\parindent0pt\paragraph{\textbf{Testing program analyzers.}}
SMT solvers are fundamental tools for various program analyzers. 
Hence, bugs in SMT solvers may affect the reliability of program analyzers.
Especially because program analyzers have become mature for practical use in recent years,
ensuring the reliability of program analyzers is crucial~\citep{symex-transf:fse-ni-13}.
There are several works on program analyzer's robustness. 
For example, Zhang \etal~\citep{zhang2019finding} tested software model checkers via reachability queries,
Bugariu \etal~\citep{bugariu2018automatically} found soundness and precision bugs in 
numerical abstract domains, Qiu \etal~\citep{qiu2018analyzing} and Pauck \etal~\citep{pauck2018android}
found bugs in the analyzers of Android apps. 
Type-aware operator mutation contributes to testing program analyzers by finding  
bugs in SMT solvers. Differential testing based approaches have been effective in finding bugs in program analyzers.
For example, Klinger \etal~\citep{klinger2019differentially} and Kapus \etal~\citep{klinger2019differentially} proposed 
the approaches for testing software model checkers and symbolic executors respectively using differential testing,
Wu \etal~\citep{wu2013effective} tested alias analyzers by cross-checking the dynamic aliasing information.
In this work, \toolname leverages differential testing to detect soundness bugs in SMT solvers. 
}

{\parindent0pt\paragraph{\textbf{Mutation-based testing.}}
Type-aware operator mutation belongs to the family of mutation-based testing techniques. 
The closest work is skeletal program enumeration (SPE) \citep{zhang2017skeletal}, 
an approach for validating compilers. Similar to type-aware mutation testing, 
program skeletons are generated from a set of seed programs. The holes in these  
skeletons are then systematically filled by exhaustive enumeration.
However, unlike type-aware operator mutation, SPE focuses on program variables 
and not on functions. SPE provides relative guarantees with respect to the input 
seed programs. 
Type-aware operator mutation is also related to FuzzChick~\cite{lampropoulos-et-al-oopsla19},   
a coverage-guided fuzzer for Coq programs. FuzzChick generates test cases   
by semantic mutations at type-level. FuzzChick is aware of parameter 
types and generates new values for the parameters while preserving type-correctness. 
Type-aware operator mutation, on the other hand, is focusing on the operators' 
types and to generate highly diverse SMT formulas. 
}

Type-aware operator mutation also belongs to black-box fuzzing techniques.
The black-box fuzzing techniques, such as SYMFUZZ~\citep{cha2015program}, 
leverage user-provided seeds and generate new mutated inputs to uncover security issues.
Grey-box fuzzing enhances black-box fuzzing by code coverage guidance and
has been successfully applied to software testing recently.
AFL~\citep{afl} is a popular tool for binary grey-box fuzzing. 
Follow-up works, such as FairFuzz~\citep{lemieux2018fairfuzz} and Steelix~\citep{li2017steelix}, improved the performance of AFL on the binary level.
However, binary level fuzzing is ineffective on programs with highly structured  
inputs (\eg. PDF viewers, programming language engines \etc) because of the many syntactically    
invalid inputs being generated.  
To generate valid test inputs, grammar-aware grey-box fuzzers were proposed.  
AFLSmart~\citep{pham2019smart}, Superion~\citep{wang-et-al-icse2019} and Nautilus~\citep{aschermann2019nautilus} 
are general grammar-aware grey-box fuzzers targeting programming language engines.
They use code coverage to guide the grammar-aware mutations. 
As a key difference to type-aware operator mutation, 
they both need to fully parse the program and work on the abstract syntax tree level,
which may lead to a higher computational cost during fuzzing. 
Type-aware operator mutation, on the other hand, works on the token level and without
fully parsing the formula.

Besides general black-box and grey-box fuzzing, various domain-specific fuzzing approaches exist, e.g., 
for testing compilers~\citep{zhang2017skeletal, yang2011finding, le2014compiler, lidbury2015many, cummins2018compiler},
testing database management systems~\citep{jung2019apollo, sqlsmith, mishra2008generating, manuel2020decting},
and testing OS kernel~\citep{corina2017difuze,han2017imf,schumilo2017kafl}. 
Type-aware operator mutation is also a domain-specific fuzzing technique 
which is unusually effective for testing SMT solvers. 


%% file: sections/conclusion.tex
\section{Conclusion}
\label{sec:conclusion}
We introduced \approachname, a simple and effective approach for stress-testing 
SMT solvers. We realized  \approachname in our testing tool \toolname in little 
more than $200$ LoC, supporting only the most basic operators of the SMT-LIB 
language. Despite this, \toolname found $\confirmed$ confirmed unique bugs (\fixed fixed) 
in Z3 and CVC4. These bug findings are highly diverse, ranging over various types, 
logics and solver configurations in both state-of-the-art SMT solvers. Among these ones, there were many critical bugs. Type-aware operator mutation has found many more bugs than previous 
approaches by a large margin. Our bug findings show that SMT solvers are not yet 
reliable enough, even the most popular and stable, such as Z3 and CVC4. Our highly practical tool 
\toolname can help SMT solver developers making their solvers more reliable. 
For future work, we want to explore the full potential of \approachname by invoking 
more sophisticated type-aware mutations. 

%% file: sections/appendix.tex
\appendix
\section{Test Seed Formulas}
\label{appendix}
\renewcommand{\arraystretch}{1.0}

\begin{figure}[h]
    {\scriptsize
    \begin{subfigure}[t]{.4\textwidth}
        \seedtableOne
        \vspace{-1cm}
    \end{subfigure}
    \hspace{0.8cm}
    \begin{subfigure}[t]{.4\textwidth}
        \seedtableTwo
    \end{subfigure}
    }
\caption{Formula counts for the respective benchmarks sets. 
    Colum \#non-inc refers to the count of non-incremental SMT-LIB files,     
    colum \#inc refers to the count of incremental SMT-LIB files. 
    \emph{z3test} and \emph{cvc4regr} 
    refer to CVC4's and Z3's respective regression test suites.\label{fig:seedtable}
}
\end{figure}

%% file: main.bbl

\begin{thebibliography}{54}


\ifx \showCODEN    \undefined \def \showCODEN     #1{\unskip}     \fi
\ifx \showDOI      \undefined \def \showDOI       #1{#1}\fi
\ifx \showISBNx    \undefined \def \showISBNx     #1{\unskip}     \fi
\ifx \showISBNxiii \undefined \def \showISBNxiii  #1{\unskip}     \fi
\ifx \showISSN     \undefined \def \showISSN      #1{\unskip}     \fi
\ifx \showLCCN     \undefined \def \showLCCN      #1{\unskip}     \fi
\ifx \shownote     \undefined \def \shownote      #1{#1}          \fi
\ifx \showarticletitle \undefined \def \showarticletitle #1{#1}   \fi
\ifx \showURL      \undefined \def \showURL       {\relax}        \fi
\providecommand\bibfield[2]{#2}
\providecommand\bibinfo[2]{#2}
\providecommand\natexlab[1]{#1}
\providecommand\showeprint[2][]{arXiv:#2}

\bibitem[\protect\citeauthoryear{Aschermann, Frassetto, Holz, Jauernig,
  Sadeghi, and Teuchert}{Aschermann et~al\mbox{.}}{2019}]%
        {aschermann2019nautilus}
\bibfield{author}{\bibinfo{person}{Cornelius Aschermann},
  \bibinfo{person}{Tommaso Frassetto}, \bibinfo{person}{Thorsten Holz},
  \bibinfo{person}{Patrick Jauernig}, \bibinfo{person}{Ahmad-Reza Sadeghi},
  {and} \bibinfo{person}{Daniel Teuchert}.} \bibinfo{year}{2019}\natexlab{}.
\newblock \showarticletitle{NAUTILUS: Fishing for deep bugs with grammars}. In
  \bibinfo{booktitle}{\emph{NDSS}}.
\newblock


\bibitem[\protect\citeauthoryear{Barrett, Conway, Deters, Hadarean,
  Jovanovi\'{c}, King, Reynolds, and Tinelli}{Barrett et~al\mbox{.}}{2011}]%
        {barrett-cav11}
\bibfield{author}{\bibinfo{person}{Clark Barrett},
  \bibinfo{person}{Christopher~L. Conway}, \bibinfo{person}{Morgan Deters},
  \bibinfo{person}{Liana Hadarean}, \bibinfo{person}{Dejan Jovanovi\'{c}},
  \bibinfo{person}{Tim King}, \bibinfo{person}{Andrew Reynolds}, {and}
  \bibinfo{person}{Cesare Tinelli}.} \bibinfo{year}{2011}\natexlab{}.
\newblock \showarticletitle{{CVC4}}. In \bibinfo{booktitle}{\emph{CAV}}.
  \bibinfo{pages}{171--177}.
\newblock


\bibitem[\protect\citeauthoryear{Barrett, Fontaine, and Tinelli}{Barrett
  et~al\mbox{.}}{2020}]%
        {SMTLIB}
\bibfield{author}{\bibinfo{person}{Clark Barrett}, \bibinfo{person}{Pascal
  Fontaine}, {and} \bibinfo{person}{Cesare Tinelli}.}
  \bibinfo{year}{2020}\natexlab{}.
\newblock \bibinfo{booktitle}{\emph{The satisfiability modulo theories library
  ({SMT-LIB})}}.
\newblock
\urldef\tempurl%
\url{{www.SMT-LIB.org}}
\showURL{%
Retrieved 2020-05-15 from \tempurl}


\bibitem[\protect\citeauthoryear{Barrett, Stump, and Tinelli}{Barrett
  et~al\mbox{.}}{2010}]%
        {BarST-SMT-10}
\bibfield{author}{\bibinfo{person}{Clark Barrett}, \bibinfo{person}{Aaron
  Stump}, {and} \bibinfo{person}{Cesare Tinelli}.}
  \bibinfo{year}{2010}\natexlab{}.
\newblock \showarticletitle{{The SMT-LIB standard: Version 2.0}}. In
  \bibinfo{booktitle}{\emph{SMT}}.
\newblock


\bibitem[\protect\citeauthoryear{Berzish, Zheng, and Ganesh}{Berzish
  et~al\mbox{.}}{2017}]%
        {berzish-et-al-2017}
\bibfield{author}{\bibinfo{person}{Murphy Berzish}, \bibinfo{person}{Yunhui
  Zheng}, {and} \bibinfo{person}{Vijay Ganesh}.}
  \bibinfo{year}{2017}\natexlab{}.
\newblock \showarticletitle{Z3str3: A string solver with theory-aware
  branching}.
\newblock \bibinfo{journal}{\emph{FMCAD}} (\bibinfo{year}{2017}),
  \bibinfo{pages}{55--59}.
\newblock


\bibitem[\protect\citeauthoryear{Blotsky, Mora, Berzish, Zheng, Kabir, and
  Ganesh}{Blotsky et~al\mbox{.}}{2018}]%
        {blotsky-et-al-cav18}
\bibfield{author}{\bibinfo{person}{Dmitry Blotsky}, \bibinfo{person}{Federico
  Mora}, \bibinfo{person}{Murphy Berzish}, \bibinfo{person}{Yunhui Zheng},
  \bibinfo{person}{Ifaz Kabir}, {and} \bibinfo{person}{Vijay Ganesh}.}
  \bibinfo{year}{2018}\natexlab{}.
\newblock \showarticletitle{{StringFuzz}: A fuzzer for string solvers}. In
  \bibinfo{booktitle}{\emph{CAV}}. \bibinfo{pages}{45--51}.
\newblock


\bibitem[\protect\citeauthoryear{Brummayer and Biere}{Brummayer and
  Biere}{2009a}]%
        {BiereBrummayer-TACAS09}
\bibfield{author}{\bibinfo{person}{Robert Brummayer} {and}
  \bibinfo{person}{Armin Biere}.} \bibinfo{year}{2009}\natexlab{a}.
\newblock \showarticletitle{Boolector: An efficient {SMT} solver for
  bit-vectors and arrays}. In \bibinfo{booktitle}{\emph{TACAS}}.
  \bibinfo{pages}{174--177}.
\newblock


\bibitem[\protect\citeauthoryear{Brummayer and Biere}{Brummayer and
  Biere}{2009b}]%
        {brummayer2009fuzzing}
\bibfield{author}{\bibinfo{person}{Robert Brummayer} {and}
  \bibinfo{person}{Armin Biere}.} \bibinfo{year}{2009}\natexlab{b}.
\newblock \showarticletitle{Fuzzing and delta-debugging {SMT} solvers}. In
  \bibinfo{booktitle}{\emph{SMT}}. \bibinfo{pages}{1--5}.
\newblock


\bibitem[\protect\citeauthoryear{Bugariu and M{\"u}ller}{Bugariu and
  M{\"u}ller}{2020}]%
        {bugariu2020automatically}
\bibfield{author}{\bibinfo{person}{Alexandra Bugariu} {and}
  \bibinfo{person}{Peter M{\"u}ller}.} \bibinfo{year}{2020}\natexlab{}.
\newblock \showarticletitle{Automatically testing string solvers}. In
  \bibinfo{booktitle}{\emph{ICSE}}. \bibinfo{pages}{459--1470}.
\newblock


\bibitem[\protect\citeauthoryear{Bugariu, W{\"u}stholz, Christakis, and
  M{\"u}ller}{Bugariu et~al\mbox{.}}{2018}]%
        {bugariu2018automatically}
\bibfield{author}{\bibinfo{person}{Alexandra Bugariu},
  \bibinfo{person}{Valentin W{\"u}stholz}, \bibinfo{person}{Maria Christakis},
  {and} \bibinfo{person}{Peter M{\"u}ller}.} \bibinfo{year}{2018}\natexlab{}.
\newblock \showarticletitle{Automatically testing implementations of numerical
  abstract domains}. In \bibinfo{booktitle}{\emph{ASE}}.
  \bibinfo{pages}{768--778}.
\newblock


\bibitem[\protect\citeauthoryear{Cadar and Donaldson}{Cadar and
  Donaldson}{2016}]%
        {symex-transf:fse-ni-13}
\bibfield{author}{\bibinfo{person}{Cristian Cadar} {and}
  \bibinfo{person}{Alastair Donaldson}.} \bibinfo{year}{2016}\natexlab{}.
\newblock \showarticletitle{Analysing the program analyser}. In
  \bibinfo{booktitle}{\emph{ICSE}}. \bibinfo{pages}{765--768}.
\newblock


\bibitem[\protect\citeauthoryear{Cadar, Dunbar, and Engler}{Cadar
  et~al\mbox{.}}{2008}]%
        {cadar-et-al2008}
\bibfield{author}{\bibinfo{person}{Cristian Cadar}, \bibinfo{person}{Daniel
  Dunbar}, {and} \bibinfo{person}{Dawson~R. Engler}.}
  \bibinfo{year}{2008}\natexlab{}.
\newblock \showarticletitle{{KLEE:} Unassisted and automatic generation of
  high-coverage tests for complex systems programs}. In
  \bibinfo{booktitle}{\emph{OSDI}}. \bibinfo{pages}{209--224}.
\newblock


\bibitem[\protect\citeauthoryear{Cha, Woo, and Brumley}{Cha
  et~al\mbox{.}}{2015}]%
        {cha2015program}
\bibfield{author}{\bibinfo{person}{Sang~Kil Cha}, \bibinfo{person}{Maverick
  Woo}, {and} \bibinfo{person}{David Brumley}.}
  \bibinfo{year}{2015}\natexlab{}.
\newblock \showarticletitle{Program-adaptive mutational fuzzing}. In
  \bibinfo{booktitle}{\emph{SP}}. \bibinfo{pages}{725--741}.
\newblock


\bibitem[\protect\citeauthoryear{Cimatti, Griggio, Schaafsma, and
  Sebastiani}{Cimatti et~al\mbox{.}}{2013}]%
        {cimatti-et-al2013}
\bibfield{author}{\bibinfo{person}{Alessandro Cimatti},
  \bibinfo{person}{Alberto Griggio}, \bibinfo{person}{Bastiaan Schaafsma},
  {and} \bibinfo{person}{Roberto Sebastiani}.} \bibinfo{year}{2013}\natexlab{}.
\newblock \showarticletitle{{The MathSAT5 SMT solver}}. In
  \bibinfo{booktitle}{\emph{TACAS}}. \bibinfo{pages}{93--107}.
\newblock


\bibitem[\protect\citeauthoryear{Competition.}{Competition.}{2020}]%
        {SMT-COMP}
\bibfield{author}{\bibinfo{person}{The International~SMT Competition.}}
  \bibinfo{year}{2020}\natexlab{}.
\newblock \bibinfo{booktitle}{\emph{{SMT-COMP}}}.
\newblock
\urldef\tempurl%
\url{https://smt-comp.github.io/2019/index.html}
\showURL{%
Retrieved 2020-05-15 from \tempurl}


\bibitem[\protect\citeauthoryear{Corina, Machiry, Salls, Shoshitaishvili, Hao,
  Kruegel, and Vigna}{Corina et~al\mbox{.}}{2017}]%
        {corina2017difuze}
\bibfield{author}{\bibinfo{person}{Jake Corina}, \bibinfo{person}{Aravind
  Machiry}, \bibinfo{person}{Christopher Salls}, \bibinfo{person}{Yan
  Shoshitaishvili}, \bibinfo{person}{Shuang Hao}, \bibinfo{person}{Christopher
  Kruegel}, {and} \bibinfo{person}{Giovanni Vigna}.}
  \bibinfo{year}{2017}\natexlab{}.
\newblock \showarticletitle{Difuze: Interface aware fuzzing for kernel
  drivers}. In \bibinfo{booktitle}{\emph{CCS}}. \bibinfo{pages}{2123--2138}.
\newblock


\bibitem[\protect\citeauthoryear{Cummins, Petoumenos, Murray, and
  Leather}{Cummins et~al\mbox{.}}{2018}]%
        {cummins2018compiler}
\bibfield{author}{\bibinfo{person}{Chris Cummins}, \bibinfo{person}{Pavlos
  Petoumenos}, \bibinfo{person}{Alastair Murray}, {and} \bibinfo{person}{Hugh
  Leather}.} \bibinfo{year}{2018}\natexlab{}.
\newblock \showarticletitle{Compiler fuzzing through deep learning}. In
  \bibinfo{booktitle}{\emph{ISSTA}}. \bibinfo{pages}{95--105}.
\newblock


\bibitem[\protect\citeauthoryear{CVC4}{CVC4}{2020}]%
        {CVC4test}
\bibfield{author}{\bibinfo{person}{CVC4}.} \bibinfo{year}{2020}\natexlab{}.
\newblock \bibinfo{booktitle}{\emph{{CVC4 Regression Test Suite}}}.
\newblock
\urldef\tempurl%
\url{https://github.com/CVC4/CVC4/tree/master/test/regress}
\showURL{%
Retrieved 2020-05-15 from \tempurl}


\bibitem[\protect\citeauthoryear{de~Moura and Bj{\o}rner}{de~Moura and
  Bj{\o}rner}{2008}]%
        {moura-bjoerner-tacas08}
\bibfield{author}{\bibinfo{person}{Leonardo de Moura} {and}
  \bibinfo{person}{Nikolaj Bj{\o}rner}.} \bibinfo{year}{2008}\natexlab{}.
\newblock \showarticletitle{Z3: An efficient {SMT} solver}. In
  \bibinfo{booktitle}{\emph{TACAS}}. \bibinfo{pages}{337--340}.
\newblock


\bibitem[\protect\citeauthoryear{DeLine and Leino}{DeLine and Leino}{2005}]%
        {deline-et-al2005}
\bibfield{author}{\bibinfo{person}{Rob DeLine} {and} \bibinfo{person}{Rustan
  Leino}.} \bibinfo{year}{2005}\natexlab{}.
\newblock \bibinfo{booktitle}{\emph{BoogiePL: A typed procedural language for
  checking object-oriented programs}}.
\newblock \bibinfo{type}{{T}echnical {R}eport}.
\newblock


\bibitem[\protect\citeauthoryear{Detlefs, Nelson, and Saxe}{Detlefs
  et~al\mbox{.}}{2005}]%
        {detlefs-et-al2005}
\bibfield{author}{\bibinfo{person}{David Detlefs}, \bibinfo{person}{Greg
  Nelson}, {and} \bibinfo{person}{James~B. Saxe}.}
  \bibinfo{year}{2005}\natexlab{}.
\newblock \showarticletitle{Simplify: A theorem prover for program checking}.
\newblock \bibinfo{journal}{\emph{JACM}} (\bibinfo{year}{2005}),
  \bibinfo{pages}{365--473}.
\newblock


\bibitem[\protect\citeauthoryear{Godefroid, Klarlund, and Sen}{Godefroid
  et~al\mbox{.}}{2005}]%
        {godefroid-et-al-pldi-2005}
\bibfield{author}{\bibinfo{person}{Patrice Godefroid}, \bibinfo{person}{Nils
  Klarlund}, {and} \bibinfo{person}{Koushik Sen}.}
  \bibinfo{year}{2005}\natexlab{}.
\newblock \showarticletitle{{DART:} Directed automated random testing}. In
  \bibinfo{booktitle}{\emph{PLDI}}. \bibinfo{pages}{213--223}.
\newblock


\bibitem[\protect\citeauthoryear{Han and Cha}{Han and Cha}{2017}]%
        {han2017imf}
\bibfield{author}{\bibinfo{person}{HyungSeok Han} {and}
  \bibinfo{person}{Sang~Kil Cha}.} \bibinfo{year}{2017}\natexlab{}.
\newblock \showarticletitle{Imf: Inferred model-based fuzzer}. In
  \bibinfo{booktitle}{\emph{CCS}}. \bibinfo{pages}{2345--2358}.
\newblock


\bibitem[\protect\citeauthoryear{Jung, Hu, Arulraj, Kim, and Kang}{Jung
  et~al\mbox{.}}{2019}]%
        {jung2019apollo}
\bibfield{author}{\bibinfo{person}{Jinho Jung}, \bibinfo{person}{Hong Hu},
  \bibinfo{person}{Joy Arulraj}, \bibinfo{person}{Taesoo Kim}, {and}
  \bibinfo{person}{Woonhak Kang}.} \bibinfo{year}{2019}\natexlab{}.
\newblock \showarticletitle{APOLLO: Automatic detection and diagnosis of
  performance regressions in database systems}. In
  \bibinfo{booktitle}{\emph{VLDB}}. \bibinfo{pages}{57--70}.
\newblock


\bibitem[\protect\citeauthoryear{Klinger, Christakis, and W{\"u}stholz}{Klinger
  et~al\mbox{.}}{2019}]%
        {klinger2019differentially}
\bibfield{author}{\bibinfo{person}{Christian Klinger}, \bibinfo{person}{Maria
  Christakis}, {and} \bibinfo{person}{Valentin W{\"u}stholz}.}
  \bibinfo{year}{2019}\natexlab{}.
\newblock \showarticletitle{Differentially testing soundness and precision of
  program analyzers}. In \bibinfo{booktitle}{\emph{ISSTA}}.
  \bibinfo{pages}{239--250}.
\newblock


\bibitem[\protect\citeauthoryear{Lampropoulos, Hicks, and Pierce}{Lampropoulos
  et~al\mbox{.}}{2019}]%
        {lampropoulos-et-al-oopsla19}
\bibfield{author}{\bibinfo{person}{Leonidas Lampropoulos},
  \bibinfo{person}{Michael Hicks}, {and} \bibinfo{person}{Benjamin~C. Pierce}.}
  \bibinfo{year}{2019}\natexlab{}.
\newblock \showarticletitle{Coverage guided, property based testing}. In
  \bibinfo{booktitle}{\emph{OOPSLA}}. \bibinfo{pages}{181:1--181:29}.
\newblock


\bibitem[\protect\citeauthoryear{Le, Afshari, and Su}{Le et~al\mbox{.}}{2014}]%
        {le2014compiler}
\bibfield{author}{\bibinfo{person}{Vu Le}, \bibinfo{person}{Mehrdad Afshari},
  {and} \bibinfo{person}{Zhendong Su}.} \bibinfo{year}{2014}\natexlab{}.
\newblock \showarticletitle{Compiler validation via equivalence modulo inputs}.
  In \bibinfo{booktitle}{\emph{PLDI}}. \bibinfo{pages}{216--226}.
\newblock


\bibitem[\protect\citeauthoryear{Lemieux and Sen}{Lemieux and Sen}{2018}]%
        {lemieux2018fairfuzz}
\bibfield{author}{\bibinfo{person}{Caroline Lemieux} {and}
  \bibinfo{person}{Koushik Sen}.} \bibinfo{year}{2018}\natexlab{}.
\newblock \showarticletitle{Fairfuzz: A targeted mutation strategy for
  increasing greybox fuzz testing coverage}. In
  \bibinfo{booktitle}{\emph{ASE}}.
\newblock


\bibitem[\protect\citeauthoryear{Li, Chen, Chandramohan, Lin, Liu, and Tiu}{Li
  et~al\mbox{.}}{2017}]%
        {li2017steelix}
\bibfield{author}{\bibinfo{person}{Yuekang Li}, \bibinfo{person}{Bihuan Chen},
  \bibinfo{person}{Mahinthan Chandramohan}, \bibinfo{person}{Shang-Wei Lin},
  \bibinfo{person}{Yang Liu}, {and} \bibinfo{person}{Alwen Tiu}.}
  \bibinfo{year}{2017}\natexlab{}.
\newblock \showarticletitle{Steelix: Program-state based binary fuzzing}. In
  \bibinfo{booktitle}{\emph{ESEC/FSE}}.
\newblock


\bibitem[\protect\citeauthoryear{Lidbury, Lascu, Chong, and Donaldson}{Lidbury
  et~al\mbox{.}}{2015}]%
        {lidbury2015many}
\bibfield{author}{\bibinfo{person}{Christopher Lidbury},
  \bibinfo{person}{Andrei Lascu}, \bibinfo{person}{Nathan Chong}, {and}
  \bibinfo{person}{Alastair~F Donaldson}.} \bibinfo{year}{2015}\natexlab{}.
\newblock \showarticletitle{Many-core compiler fuzzing}. In
  \bibinfo{booktitle}{\emph{PLDI}}. \bibinfo{pages}{65--76}.
\newblock


\bibitem[\protect\citeauthoryear{Mishra, Koudas, and Zuzarte}{Mishra
  et~al\mbox{.}}{2008}]%
        {mishra2008generating}
\bibfield{author}{\bibinfo{person}{Chaitanya Mishra}, \bibinfo{person}{Nick
  Koudas}, {and} \bibinfo{person}{Calisto Zuzarte}.}
  \bibinfo{year}{2008}\natexlab{}.
\newblock \showarticletitle{Generating targeted queries for database testing}.
  In \bibinfo{booktitle}{\emph{SIGMOD}}. \bibinfo{pages}{499--510}.
\newblock


\bibitem[\protect\citeauthoryear{Niemetz and Biere}{Niemetz and Biere}{2013}]%
        {NiemetzBiere-SMT13}
\bibfield{author}{\bibinfo{person}{Aina Niemetz} {and} \bibinfo{person}{Armin
  Biere}.} \bibinfo{year}{2013}\natexlab{}.
\newblock \showarticletitle{ddSMT: A delta debugger for the SMT-LIB v2 format}.
  In \bibinfo{booktitle}{\emph{SMT}}. \bibinfo{pages}{36--45}.
\newblock


\bibitem[\protect\citeauthoryear{Niemetz, Preiner, and Biere}{Niemetz
  et~al\mbox{.}}{2017}]%
        {niemetz2017model}
\bibfield{author}{\bibinfo{person}{Aina Niemetz}, \bibinfo{person}{Mathias
  Preiner}, {and} \bibinfo{person}{Armin Biere}.}
  \bibinfo{year}{2017}\natexlab{}.
\newblock \showarticletitle{Model-based API testing for SMT solvers}. In
  \bibinfo{booktitle}{\emph{SMT}}. \bibinfo{pages}{10}.
\newblock


\bibitem[\protect\citeauthoryear{Numair, {Christakis}, {W{\"u}stholz}, and
  {Zhang}}{Numair et~al\mbox{.}}{2020}]%
        {mansur-etal-arxiv2020}
\bibfield{author}{\bibinfo{person}{Mansur Numair}, \bibinfo{person}{Maria
  {Christakis}}, \bibinfo{person}{Valentin {W{\"u}stholz}}, {and}
  \bibinfo{person}{Fuyuan {Zhang}}.} \bibinfo{year}{2020}\natexlab{}.
\newblock \showarticletitle{{Detecting critical bugs in SMT solvers using
  blackbox mutational fuzzing}}.
\newblock \bibinfo{journal}{\emph{arXiv e-prints}} (\bibinfo{date}{April}
  \bibinfo{year}{2020}), \bibinfo{pages}{arXiv:2004.05934}.
\newblock


\bibitem[\protect\citeauthoryear{Pauck, Bodden, and Wehrheim}{Pauck
  et~al\mbox{.}}{2018}]%
        {pauck2018android}
\bibfield{author}{\bibinfo{person}{Felix Pauck}, \bibinfo{person}{Eric Bodden},
  {and} \bibinfo{person}{Heike Wehrheim}.} \bibinfo{year}{2018}\natexlab{}.
\newblock \showarticletitle{Do {Android} taint analysis tools keep their
  promises?}. In \bibinfo{booktitle}{\emph{ESEC/FSE}}.
  \bibinfo{pages}{331--341}.
\newblock


\bibitem[\protect\citeauthoryear{Pham, B{\"o}hme, Santosa, Caciulescu, and
  Roychoudhury}{Pham et~al\mbox{.}}{2019}]%
        {pham2019smart}
\bibfield{author}{\bibinfo{person}{Van-Thuan Pham}, \bibinfo{person}{Marcel
  B{\"o}hme}, \bibinfo{person}{Andrew~Edward Santosa},
  \bibinfo{person}{Alexandru~Razvan Caciulescu}, {and} \bibinfo{person}{Abhik
  Roychoudhury}.} \bibinfo{year}{2019}\natexlab{}.
\newblock \showarticletitle{Smart greybox fuzzing}.
\newblock \bibinfo{journal}{\emph{TSE}} (\bibinfo{year}{2019}).
\newblock


\bibitem[\protect\citeauthoryear{Qiu, Wang, and Rubin}{Qiu
  et~al\mbox{.}}{2018}]%
        {qiu2018analyzing}
\bibfield{author}{\bibinfo{person}{Lina Qiu}, \bibinfo{person}{Yingying Wang},
  {and} \bibinfo{person}{Julia Rubin}.} \bibinfo{year}{2018}\natexlab{}.
\newblock \showarticletitle{Analyzing the analyzers: {FlowDroid/IccTA},
  {AmanDroid}, and {DroidSafe}}. In \bibinfo{booktitle}{\emph{ISSTA}}.
  \bibinfo{pages}{176--186}.
\newblock


\bibitem[\protect\citeauthoryear{Regehr, Chen, Cuoq, Eide, Ellison, and
  Yang}{Regehr et~al\mbox{.}}{2012}]%
        {Regehr:2012:TRC:2254064.2254104}
\bibfield{author}{\bibinfo{person}{John Regehr}, \bibinfo{person}{Yang Chen},
  \bibinfo{person}{Pascal Cuoq}, \bibinfo{person}{Eric Eide},
  \bibinfo{person}{Chucky Ellison}, {and} \bibinfo{person}{Xuejun Yang}.}
  \bibinfo{year}{2012}\natexlab{}.
\newblock \showarticletitle{Test-case reduction for {C} compiler bugs}. In
  \bibinfo{booktitle}{\emph{PLDI}}. \bibinfo{pages}{335--346}.
\newblock


\bibitem[\protect\citeauthoryear{Reynolds, Deters, Kuncak, Barrett, and
  Tinelli}{Reynolds et~al\mbox{.}}{2015}]%
        {reynolds-et-al-cav2015}
\bibfield{author}{\bibinfo{person}{Andrew Reynolds}, \bibinfo{person}{Morgan
  Deters}, \bibinfo{person}{Viktor Kuncak}, \bibinfo{person}{Clark~W. Barrett},
  {and} \bibinfo{person}{Cesare Tinelli}.} \bibinfo{year}{2015}\natexlab{}.
\newblock \showarticletitle{On counterexample guided quantifier instantiation
  for synthesis in {CVC4}}. In \bibinfo{booktitle}{\emph{CAV}}.
\newblock


\bibitem[\protect\citeauthoryear{Rigger and Su}{Rigger and Su}{2020}]%
        {manuel2020decting}
\bibfield{author}{\bibinfo{person}{Manuel Rigger} {and}
  \bibinfo{person}{Zhendong Su}.} \bibinfo{year}{2020}\natexlab{}.
\newblock \showarticletitle{Detecting optimization bugs in database engines via
  non-optimizing reference Engine Construction}. In
  \bibinfo{booktitle}{\emph{OOPSLA}}.
\newblock


\bibitem[\protect\citeauthoryear{Schumilo, Aschermann, Gawlik, Schinzel, and
  Holz}{Schumilo et~al\mbox{.}}{2017}]%
        {schumilo2017kafl}
\bibfield{author}{\bibinfo{person}{Sergej Schumilo}, \bibinfo{person}{Cornelius
  Aschermann}, \bibinfo{person}{Robert Gawlik}, \bibinfo{person}{Sebastian
  Schinzel}, {and} \bibinfo{person}{Thorsten Holz}.}
  \bibinfo{year}{2017}\natexlab{}.
\newblock \showarticletitle{kAFL: Hardware-assisted feedback fuzzing for OS
  kernels}. In \bibinfo{booktitle}{\emph{USENIX Security}}.
  \bibinfo{pages}{167--182}.
\newblock


\bibitem[\protect\citeauthoryear{Scott, Mora, and Ganesh}{Scott
  et~al\mbox{.}}{2020}]%
        {scott-et-al-cav20}
\bibfield{author}{\bibinfo{person}{Joseph Scott}, \bibinfo{person}{Federico
  Mora}, {and} \bibinfo{person}{Vijay Ganesh}.}
  \bibinfo{year}{2020}\natexlab{}.
\newblock \showarticletitle{{BanditFuzz}: Fuzzing SMT solvers with
  reinforcement learning}. In \bibinfo{booktitle}{\emph{CAV}}.
\newblock


\bibitem[\protect\citeauthoryear{Seltenreich}{Seltenreich}{2020}]%
        {sqlsmith}
\bibfield{author}{\bibinfo{person}{Andreas Seltenreich}.}
  \bibinfo{year}{2020}\natexlab{}.
\newblock \bibinfo{booktitle}{\emph{SQLSmith}}.
\newblock
\urldef\tempurl%
\url{https://github.com/anse1/sqlsmith}
\showURL{%
Retrieved 2020-08-13 from \tempurl}


\bibitem[\protect\citeauthoryear{SMT-LIB}{SMT-LIB}{2020}]%
        {SMT-LIB}
\bibfield{author}{\bibinfo{person}{SMT-LIB}.} \bibinfo{year}{2020}\natexlab{}.
\newblock \bibinfo{booktitle}{\emph{{SMT-LIB Benchmarks}}}.
\newblock
\urldef\tempurl%
\url{http://smtlib.cs.uiowa.edu/benchmarks.shtml}
\showURL{%
Retrieved 2020-05-15 from \tempurl}


\bibitem[\protect\citeauthoryear{Solar-Lezama}{Solar-Lezama}{2008}]%
        {solar2008program}
\bibfield{author}{\bibinfo{person}{Armando Solar-Lezama}.}
  \bibinfo{year}{2008}\natexlab{}.
\newblock \emph{\bibinfo{title}{Program synthesis by sketching}}.
\newblock \bibinfo{thesistype}{Ph.D. Dissertation}. \bibinfo{school}{UC
  Berkeley}.
\newblock
\urldef\tempurl%
\url{https://people.csail.mit.edu/asolar/papers/thesis.pdf}
\showURL{%
\tempurl}


\bibitem[\protect\citeauthoryear{Torlak and Bodik}{Torlak and Bodik}{2014}]%
        {torlak2014lightweight}
\bibfield{author}{\bibinfo{person}{Emina Torlak} {and}
  \bibinfo{person}{Rastislav Bodik}.} \bibinfo{year}{2014}\natexlab{}.
\newblock \showarticletitle{A lightweight symbolic virtual machine for
  solver-aided host languages}. In \bibinfo{booktitle}{\emph{PLDI}}.
  \bibinfo{pages}{530--541}.
\newblock


\bibitem[\protect\citeauthoryear{Wang, Chen, Wei, and Liu}{Wang
  et~al\mbox{.}}{2019}]%
        {wang-et-al-icse2019}
\bibfield{author}{\bibinfo{person}{Junjie Wang}, \bibinfo{person}{Bihuan Chen},
  \bibinfo{person}{Lei Wei}, {and} \bibinfo{person}{Yang Liu}.}
  \bibinfo{year}{2019}\natexlab{}.
\newblock \showarticletitle{Superion: Grammar-aware greybox fuzzing}. In
  \bibinfo{booktitle}{\emph{ICSE}}. \bibinfo{pages}{724--735}.
\newblock


\bibitem[\protect\citeauthoryear{Winterer, Zhang, and Su}{Winterer
  et~al\mbox{.}}{2020}]%
        {semantic-fusion}
\bibfield{author}{\bibinfo{person}{Dominik Winterer}, \bibinfo{person}{Chengyu
  Zhang}, {and} \bibinfo{person}{Zhendong Su}.}
  \bibinfo{year}{2020}\natexlab{}.
\newblock \showarticletitle{Validating SMT solvers via semantic fusion}. In
  \bibinfo{booktitle}{\emph{PLDI}}. \bibinfo{pages}{718--730}.
\newblock


\bibitem[\protect\citeauthoryear{Wu, Hu, Tang, and Yang}{Wu
  et~al\mbox{.}}{2013}]%
        {wu2013effective}
\bibfield{author}{\bibinfo{person}{Jingyue Wu}, \bibinfo{person}{Gang Hu},
  \bibinfo{person}{Yang Tang}, {and} \bibinfo{person}{Junfeng Yang}.}
  \bibinfo{year}{2013}\natexlab{}.
\newblock \showarticletitle{Effective dynamic detection of alias analysis
  errors}. In \bibinfo{booktitle}{\emph{ESEC/FSE}}. \bibinfo{pages}{279--289}.
\newblock


\bibitem[\protect\citeauthoryear{Yang, Chen, Eide, and Regehr}{Yang
  et~al\mbox{.}}{2011}]%
        {yang2011finding}
\bibfield{author}{\bibinfo{person}{Xuejun Yang}, \bibinfo{person}{Yang Chen},
  \bibinfo{person}{Eric Eide}, {and} \bibinfo{person}{John Regehr}.}
  \bibinfo{year}{2011}\natexlab{}.
\newblock \showarticletitle{Finding and understanding bugs in C compilers}. In
  \bibinfo{booktitle}{\emph{PLDI}}. \bibinfo{pages}{283--294}.
\newblock


\bibitem[\protect\citeauthoryear{Z3}{Z3}{2020}]%
        {Z3test}
\bibfield{author}{\bibinfo{person}{Z3}.} \bibinfo{year}{2020}\natexlab{}.
\newblock \bibinfo{booktitle}{\emph{{Z3 Regression Test Suite}}}.
\newblock
\urldef\tempurl%
\url{https://github.com/Z3Prover/z3test}
\showURL{%
Retrieved 2020-05-15 from \tempurl}


\bibitem[\protect\citeauthoryear{Zalewski}{Zalewski}{2020}]%
        {afl}
\bibfield{author}{\bibinfo{person}{Michal Zalewski}.}
  \bibinfo{year}{2020}\natexlab{}.
\newblock \bibinfo{booktitle}{\emph{american fuzzy lop}}.
\newblock
\urldef\tempurl%
\url{https://lcamtuf.coredump.cx/afl/}
\showURL{%
Retrieved 2020-08-12 from \tempurl}


\bibitem[\protect\citeauthoryear{Zhang, Su, Yan, Zhang, Pu, and Su}{Zhang
  et~al\mbox{.}}{2019}]%
        {zhang2019finding}
\bibfield{author}{\bibinfo{person}{Chengyu Zhang}, \bibinfo{person}{Ting Su},
  \bibinfo{person}{Yichen Yan}, \bibinfo{person}{Fuyuan Zhang},
  \bibinfo{person}{Geguang Pu}, {and} \bibinfo{person}{Zhendong Su}.}
  \bibinfo{year}{2019}\natexlab{}.
\newblock \showarticletitle{Finding and understanding bugs in software model
  checkers}. In \bibinfo{booktitle}{\emph{ESEC/FSE}}.
  \bibinfo{pages}{763--773}.
\newblock


\bibitem[\protect\citeauthoryear{Zhang, Sun, and Su}{Zhang
  et~al\mbox{.}}{2017}]%
        {zhang2017skeletal}
\bibfield{author}{\bibinfo{person}{Qirun Zhang}, \bibinfo{person}{Chengnian
  Sun}, {and} \bibinfo{person}{Zhendong Su}.} \bibinfo{year}{2017}\natexlab{}.
\newblock \showarticletitle{Skeletal program enumeration for rigorous compiler
  testing}. In \bibinfo{booktitle}{\emph{PLDI}}. \bibinfo{pages}{347--361}.
\newblock


\end{thebibliography}
